\documentclass[acmsmall,screen,nonacm]{acmart}
\usepackage{mathpartir}
\usepackage{array}
\usepackage{flafter}
\usepackage{placeins}
\usepackage{longtable}
\usepackage{tikz}
\usetikzlibrary{arrows.meta}

\definecolor{constructorink}{HTML}{24557A}
\definecolor{operationink}{HTML}{27645B}
\newif\ifnotationcolor
\notationcolortrue
\newcommand{\notationink}[2]{\ifnotationcolor\textcolor{#1}{#2}\else#2\fi}
\newcommand{\constructor}[1]{\mathord{\notationink{constructorink}{\mathsf{#1}}}}
\newcommand{\namedop}[1]{\operatorname{\notationink{operationink}{#1}}}

\newcommand{\Fsubtop}{\ensuremath{\mathsf{F}_{<:}^{\top}}}
\newcommand{\Ftop}{\Fsubtop}
\newcommand{\tytop}{\ensuremath{\top}}
\newcommand{\tmconst}{\constructor{Top}}

\newcommand{\sub}[3]{#1\vdash #2<:#3}
\newcommand{\asub}[3]{#1\vdash_{\mathrm a} #2<:#3}
\newcommand{\typing}[4]{#1;#2\vdash #3:#4}
\newcommand{\domof}[1]{\operatorname{dom}(#1)}
\newcommand{\forget}[1]{{#1}^{\circ}}
\newcommand{\subin}[1]{\mathrel{<:}_{#1}}
\newcommand{\Up}[2]{U_{#1}(#2)}
\newcommand{\Down}[2]{D_{#1}(#2)}
\newcommand{\Vars}[2]{V_{#1}(#2)}
\newcommand{\Basis}[3]{\mathcal{B}_{#1;#2}(#3)}
\newcommand{\basissynth}[4]{#1;#2\vdash_{\mathrm b}#3\Rightarrow #4}
\newcommand{\basischeck}[4]{#1;#2\vdash_{\mathrm b}#3\Leftarrow #4}
\newcommand{\canrel}[4]{#1;#2\vdash_{\mathrm c}#3\Rightarrow #4}
\newcommand{\genlang}[1]{\mathcal{L}(#1)}
\newcommand{\typepreorder}[1]{P_{#1}}
\newcommand{\upclosure}[2]{\mathord{\uparrow}_{#1}#2}
\newcommand{\downclosure}[2]{\mathord{\downarrow}_{#1}#2}
\newcommand{\weaksubst}{\iota}
\newcommand{\idsubst}{\namedop{id}}

\newcommand{\Can}{\mathsf{C}}
\newcommand{\Expose}[2]{{#2}^{\sharp}_{#1}}
\newcommand{\canon}[3]{\Can_{#1;#2}(#3)}
\newcommand{\checkfam}[4]{#1;#2\vdash_{\mathrm{bw}}#3\Leftarrow #4}
\newcommand{\garr}[2]{\constructor{Arr}(#1,#2)}
\newcommand{\gall}[2]{\constructor{All}(#1,#2)}
\newcommand{\qry}[4]{E^{#1}_{#2}(#3,#4)}
\newcommand{\invfam}[3]{\mathsf{I}^{#1}_{#2}(#3)}
\newcommand{\invlo}[2]{\invfam{-}{#1}{#2}}
\newcommand{\invup}[2]{\invfam{+}{#1}{#2}}
\newcommand{\invshort}[2]{\mathsf{I}^{#1}(#2)}
\newcommand{\invgate}[2]{J^{#1}(#2)}

\newcommand{\tyweight}[2]{w_{#1}(#2)}
\newcommand{\tysubst}[2]{#1[#2]}

\newcommand{\liftsub}[1]{\namedop{up}(#1)}
\newcommand{\bodysel}[4]{\mathsf{P}^{#2}_{#1,#4}(#3)}

\newcommand{\histcore}[1]{\namedop{core}(#1)}
\newcommand{\histbase}[1]{\namedop{base}(#1)}
\newcommand{\rootholes}[1]{\namedop{holes}(#1)}
\newcommand{\ancestors}[2]{\namedop{anc}_{#1}(#2)}
\newcommand{\birthrecord}[2]{\namedop{record}_{#1}(#2)}
\newcommand{\domainfam}[1]{F^{\mathrm{dom}}_{#1}}
\newcommand{\termnodes}[1]{N_{#1}}
\newcommand{\histcap}[1]{N_{#1}}
\newcommand{\spinesquares}[1]{M_2(#1)}
\newcommand{\inputledger}[2]{\namedop{ledger}(#1,#2)}
\newcommand{\runconst}{c_{\mathrm{run}}}

\newcolumntype{P}[1]{>{\raggedright\arraybackslash}p{#1}}

\newcommand{\All}[3]{\forall #1<:#2.\,#3}
\newcommand{\tabs}[3]{\Lambda #1<:#2.\,#3}
\newcommand{\abs}[3]{\lambda #1:#2.\,#3}

\DeclareRobustCommand{\ruledef}[1]{%
  \hypertarget{rule:#1}{\textnormal{\textsc{#1}}}}
\DeclareRobustCommand{\ruleref}[1]{%
  \hyperlink{rule:#1}{\textnormal{\textsc{#1}}}}
\newtheorem{theorem}{Theorem}[section]
\newtheorem{lemma}[theorem]{Lemma}
\newtheorem{definition}[theorem]{Definition}
\newtheorem{corollary}[theorem]{Corollary}

\setcopyright{none}
\makeatletter
\let\@ACM@checkaffil\relax
\let\hyxmp@parse@acmart\relax
\makeatother
\title[Typechecking without Minimum Types]{Typechecking without Minimum Types: Finite Bases and Backward Checking for \texorpdfstring{\Ftop}{F-sub-top}}
\author{Mingkun Xiao}
\affiliation{\institution{Capital Normal University}}
\email{1230604016@cnu.edu.cn}

\author{Yixuan Sun}
\affiliation{\institution{Kean University}}
\email{sunyixua@kean.edu}

\begin{document}

\begin{abstract}
System \Ftop{} has decidable subtyping but lacks minimum types. We prove
that typechecking is decidable for arbitrary explicitly annotated terms
with unrestricted subsumption. The difficulty is that type abstraction
uses a variable's declared bound to type its body, while quantified
subtyping forgets that bound. We construct finite covers that preserve
all required subtype comparisons across this change of context. These
covers yield a compositional finite typing basis and a decision procedure.

We then derive a backward checker that propagates acceptable targets
without constructing the complete typing basis. Rectangularity of typing
families gives a uniform elimination domain while retaining every result type, even when
the canonical type is not least. Finite inverse families transport
targets through type substitution. The same inverse specification
recovers bound forgetting and yields finite type-avoidance families.
Both typechecking procedures are total Rocq functions proved equivalent
to the original typing judgment. For a separate symbolic word-RAM
realization, we prove a quartic worst-case work bound. We measure family
representation sizes and comparison storage in independent OCaml prototypes
on synthetic checking judgments.
\end{abstract}

\maketitle

\section{Introduction}
\label{sec:introduction}

Decidable subtyping does not by itself provide a typechecking algorithm.
A subtype checker compares two given types. Typechecking must also choose
types for intermediate expressions, even when all parameter types and
type arguments are explicit. Subsumption permits replacing an assigned
type by a supertype before the next typing rule is applied. A computable
\emph{minimum type}, an assignable type below every other assignable type,
controls these choices: compute that type, then compare it with the target
\cite{CurienGhelli1992}.

For System \Ftop{}, Castagna and Pierce establish decidable subtyping,
but their corrigendum exhibits terms without a minimum type and leaves
general typechecking open \cite{CastagnaPierce1994,CastagnaPierce1995}.
Laird later proves decidability for $\beta$-normal terms of the original
system \cite[Corollary~8.6]{Laird2023}. A normal form may have a type that
the original term lacks, so normalizing the input does not extend this
result to arbitrary terms (Section~\ref{sec:original-terms}).

We prove that typechecking is decidable for arbitrary explicitly annotated
terms of the original \Ftop{} calculus. We compute a finite collection of
assignable types such that every assignable target is above at least one member.
We then derive a backward checker that uses the requested target to avoid
constructing the complete collection.

\paragraph{Why one type is not enough.}
In \Ftop{}, a type abstraction's body can use the declared bound of its
type variable. Comparing two quantified types instead treats that variable
as having bound $\tytop$. To see what this changes, take an outer
variable $\beta<:\tytop$ and the term
\[
  \tabs{\alpha}{\beta}{\abs{x}{\alpha}{x}}.
\]
The inner function can return $x$ at type $\alpha$ or, using the bound
$\alpha<:\beta$, at type $\beta$. Thus the term has both types
$\All{\alpha}{\beta}{\alpha\to \alpha}$ and
$\All{\alpha}{\beta}{\alpha\to \beta}$. Comparing these quantified types
forgets the bound that related their bodies. Ghelli's example rules out
a common subtype of the two types in the outer context
\cite{CastagnaPierce1995}, so no minimum type can cover both.
Section~\ref{sec:obstacles} gives the inversion argument.

\paragraph{Decidability from finite representatives.}
Keeping both body types preserves the two typings in this example.
For a general abstraction, we need enough assignable body types that,
after its bound is forgotten for comparison, every type of the abstraction
remains above one of the resulting quantified candidates. The corrigendum
suggests using finite sets of minimal types; what is needed is a
construction that covers all the choices subsumption permits.

We prove that the required finite covers can be computed for every
well-formed type and context. They preserve the necessary comparisons
after bound forgetting. For instance, to cover a variable's supertypes,
we retain the variable and representatives obtained from its bound.
The full construction combines these alternatives according to the type
constructors. Applying it at type abstraction
lets us compute a finite collection for each term from the collections
for its subterms. We call the result a \emph{typing basis}.
Write $\Delta$ for type-variable bounds and $\Gamma$ for term-variable
typings. For a term $s$ in these contexts, the computed basis satisfies
\begin{equation}
  \typing{\Delta}{\Gamma}{s}{T}
  \quad\Longleftrightarrow\quad
  \exists B\in\Basis{\Delta}{\Gamma}{s}.\ \sub{\Delta}{B}{T}.
  \label{eq:typing-basis}
\end{equation}
Every candidate is assignable to $s$, and an empty basis records failure
to type the term. Checking a requested target therefore requires only
finitely many subtype tests. The candidates need not be minimal.

\paragraph{Checking from a target.}
A finite basis establishes decidability, but its construction must retain
enough candidates to cover every assignable type. Constructing and
combining these candidates can require substantial work, even when a
checking request asks about only one target. This motivates a backward
checker that uses the requested target to determine which requirements
to pass to subterms, without first constructing the complete basis.

At an application $f\,u$, the target $T$ specifies the required result
type but leaves the parameter type undetermined. We must choose a type
$A$ such that $f$ can be checked at $A\to T$ and $u$ at $A$.
The challenge is to make this choice without losing any valid typing.

We show that whenever $f$ can be used as a function, one computable
parameter type is sufficient. It is a supertype of every admissible
parameter type. Any result type available at one of those parameter types
remains available at the chosen one.
The checker can therefore fix the argument check while leaving all result
choices available. We call this separation \emph{rectangularity}.
It holds even for functions without a minimum type. At type application,
a corresponding property fixes the bound against which to check the
explicit type argument.

At a type application, the desired result constrains the body of the
function's quantified type: after the explicit type argument is
substituted, that body must be a subtype of the target. We represent
these acceptable body types by a finite cover, allowing the requirement
to be passed back to the function. We call this operation \emph{finite
inverse substitution}. Together, the fixed parameter types and these
finite covers give a complete checker by recursion on the original term.

The inverse construction also recovers the bound-forgetting covers used
in the decidability proof. Applied to a change of variable scope, it
computes finite covers of the supertypes and subtypes expressible without
a local type variable. Section~\ref{sec:method-scope} develops this
type-avoidance result and the requirements for transferring the method
to another calculus.

The contributions are:
\begin{itemize}
  \item \emph{Decidable typechecking without minimum types.}
    We construct effective finite covers under bound forgetting and use them
    to compute a complete typing basis. This decides the original judgment
    for arbitrary explicitly annotated terms, including $\beta$-redexes
    and unrestricted subsumption
    (Sections~\ref{sec:cones} and~\ref{sec:typing-bases}).
  \item \emph{Complete backward checking by inverse substitution.}
    We prove rectangularity of typing families and exact finite inverse
    substitution, yielding a complete checker that propagates targets
    without constructing the full typing basis (Section~\ref{sec:backward}).
    Shared representations support queries without enumeration
    (Section~\ref{sec:symbolic-implementation}); a separate word-RAM realization
    has a quartic charged-work bound (Appendix~\ref{app:cost}).
\end{itemize}
The two semantic checkers are total Rocq functions proved equivalent to the
original declarative judgment. Section~\ref{sec:mechanization} describes
their verification and evaluates handwritten OCaml prototypes on complete
synthetic checking judgments. The experiments examine symbolic family sizes
and comparison storage.

\section{The System \texorpdfstring{\Ftop}{F-sub-top} Calculus}
\label{sec:calculus}

\subsection{Syntax and rules}

We use the subtype rules of Section~2 and the typing rules of Appendix~A.1
of Castagna and Pierce's extended version dated January~10, 1994
\cite{CastagnaPierce1994}, with \ruleref{All-Top} as the quantified subtyping rule.
We use lowercase Greek letters $\alpha,\beta,\gamma$ for type variables,
Latin capitals $A,B,C$ for arbitrary types, and lowercase Latin letters
$x,y,z$ for term variables. The syntax is
\begin{align*}
  A,B &::= \alpha\mid\tytop\mid A\to B\mid\All{\alpha}{A}{B},\\
  s,t &::= x\mid\tmconst\mid\abs{x}{A}{s}\mid s\,t
       \mid\tabs{\alpha}{A}{s}\mid s\{A\}.
\end{align*}
Type and term contexts are
\[
 \Delta::=\varnothing\mid\Delta,\alpha<:A,
 \qquad
 \Gamma::=\varnothing\mid\Gamma,x:A.
\]
Contexts use distinct names; bounds in $\Delta$ refer only to earlier
entries, and types in $\Gamma$ are scoped over $\Delta$. Type formation
is structural and depends only on scope. We use standard lookup notation, 
$\alpha$-equivalence and capture-avoiding substitution $B[R/\alpha]$.
Weakening of named types and term contexts is implicit.
All rules below operate on well-formed contexts and types.
We abbreviate $\sub{\Delta}{A}{B}$ by $A\subin{\Delta} B$ when convenient.

\begin{figure}[htbp]
\begin{mathpar}
  \inferrule*[right={\ruledef{Refl}}]{ }{\sub{\Delta}{A}{A}}
  \and
  \inferrule*[right={\ruledef{Trans}}]{\sub{\Delta}{A}{B}\\\sub{\Delta}{B}{C}}
    {\sub{\Delta}{A}{C}}
  \and
  \inferrule*[right={\ruledef{Bound}}]{\alpha<:A\in\Delta}{\sub{\Delta}{\alpha}{A}}
  \and
  \inferrule*[right={\ruledef{Top}}]{ }{\sub{\Delta}{A}{\tytop}}
  \and
  \inferrule*[right={\ruledef{Arrow}}]{\sub{\Delta}{C}{A}\\\sub{\Delta}{B}{D}}
    {\sub{\Delta}{A\to B}{C\to D}}
  \and
  \inferrule*[right={\ruledef{All-Top}}]
    {\sub{\Delta}{C}{A}\\\sub{\Delta,\alpha<:\tytop}{B}{D}}
    {\sub{\Delta}{\All{\alpha}{A}{B}}{\All{\alpha}{C}{D}}}
\end{mathpar}
\caption{Declarative subtyping. Formation of all displayed judgments is implicit.}
\Description{The six subtyping rules are reflexivity, transitivity, variable
bounds, the top type, contravariant arrow domains, and universal comparison
with a variable bounded by the top type in the body premise.}
\label{fig:subtyping}
\end{figure}

\begin{figure}[htbp]
\begin{mathpar}
  \inferrule*[right={\ruledef{Var}}]{x:A\in\Gamma}{\typing{\Delta}{\Gamma}{x}{A}}
  \and
  \inferrule*[right={\ruledef{Const}}]{ }{\typing{\Delta}{\Gamma}{\tmconst}{\tytop}}
  \and
  \inferrule*[right={\ruledef{Abs}}]{\typing{\Delta}{\Gamma,x:A}{s}{B}}
    {\typing{\Delta}{\Gamma}{\abs{x}{A}{s}}{A\to B}}
  \and
  \inferrule*[right={\ruledef{App}}]{\typing{\Delta}{\Gamma}{f}{A\to B}\\
    \typing{\Delta}{\Gamma}{u}{A}}{\typing{\Delta}{\Gamma}{f\,u}{B}}
  \and
  \inferrule*[right={\ruledef{Type-Abs}}]{\typing{\Delta,\alpha<:A}{\Gamma}{s}{B}}
    {\typing{\Delta}{\Gamma}{\tabs{\alpha}{A}{s}}{\All{\alpha}{A}{B}}}
  \and
  \inferrule*[right={\ruledef{Type-App}}]{\typing{\Delta}{\Gamma}{f}{\All{\alpha}{A}{B}}\\
    \sub{\Delta}{R}{A}}{\typing{\Delta}{\Gamma}{f\{R\}}{B[R/\alpha]}}
  \and
  \inferrule*[right={\ruledef{Sub}}]{\typing{\Delta}{\Gamma}{s}{A}\\\sub{\Delta}{A}{B}}
    {\typing{\Delta}{\Gamma}{s}{B}}
\end{mathpar}
\caption{Typing with unrestricted subsumption. \ruleref{Type-Abs} uses the actual
  annotation $A$, whereas \ruleref{All-Top} compares bodies with bound $\tytop$.}
\Description{The seven typing rules cover variables, the term constant Top,
function and type abstraction, their two application forms, and unrestricted
subsumption. Type abstraction extends the context with the declared bound.}
\label{fig:typing}
\end{figure}

Figures~\ref{fig:subtyping} and~\ref{fig:typing} define the judgment to be
decided. Subsumption may occur at any point in a derivation. The checker
receives $\Delta,\Gamma,s,T$ and decides whether
$\typing{\Delta}{\Gamma}{s}{T}$ holds. Type inference for unannotated terms
is outside this specification. A raw term follows the displayed grammar
but may contain unbound variables or ill-scoped annotations.

\subsection{Failure of minimum typing}
\label{sec:obstacles}

Return to Ghelli's example from Section~\ref{sec:introduction}
\cite{CastagnaPierce1995}. With $\Delta=(\beta<:\tytop)$, write
\[
  e=\tabs{\alpha}{\beta}{\abs{x}{\alpha}{x}}.
\]
Denote the two types from the introduction by
\[
  T_1=\All{\alpha}{\beta}{\alpha\to \alpha},\qquad
  T_2=\All{\alpha}{\beta}{\alpha\to \beta}.
\]
A minimum, or least, type for $e$ would have to be a common subtype of
$T_1$ and $T_2$. Subtype inversion (Section~\ref{sec:subtyping}) rules out
$\tytop$ and arrows as subtypes of a universal. The only free type variable,
$\beta$, cannot be one either, because its bound is $\tytop$.
A common subtype must therefore be universal. Comparing its body with
both target bodies uses $\beta<:\tytop,\alpha<:\tytop$; inversion then
forces an arrow body whose result lies below both $\alpha$ and $\beta$.
But the only subtype of $\alpha$ in that context is $\alpha$, and the
only subtype of $\beta$ is $\beta$. Hence no common subtype exists,
and $e$ has no minimum type.

For checking, the obstruction is the order of \ruleref{Sub} and
\ruleref{Type-Abs}. A derivation of $T_2$ first applies subsumption to
the body $\abs{x}{\alpha}{x}$, using the actual bound:
\[
 \sub{\Delta,\alpha<:\beta}{\alpha\to\alpha}{\alpha\to\beta}.
\]
The body therefore has type $\alpha\to\beta$. Applying \ruleref{Type-Abs} then
gives $\typing{\Delta}{\varnothing}{e}{T_2}$.

Alternatively, apply \ruleref{Type-Abs} directly to the body's type
$\alpha\to\alpha$, obtaining $\typing{\Delta}{\varnothing}{e}{T_1}$.
A final use of \ruleref{Sub} could reach $T_2$ only if $\sub{\Delta}{T_1}{T_2}$.
\ruleref{All-Top} compares the bodies under $\alpha<:\tytop$, where
\[
 \Delta,\alpha<:\tytop\nvdash\alpha\to\alpha<:\alpha\to\beta.
\]
This comparison would require $\alpha<:\beta$, which is unavailable in
that context. Thus the second route fails: subsumption inside the body
cannot always be postponed until after type abstraction.

The abstraction step must therefore account for typings obtained under
the actual bound before using the weaker comparison context. Keeping
both body types $\alpha\to\alpha$ and $\alpha\to\beta$ repairs this example.
For a general abstraction, each retained body type needs finite
alternatives that are assignable under the actual bound and cover its
supertypes under the comparison context with bound $\tytop$.
Section~\ref{sec:cones} formulates these two requirements and constructs
finite families satisfying them.

\subsection{Decidable subtyping and exposure}
\label{sec:subtyping}

We decide subtype comparisons using Castagna and Pierce's algorithmic relation
\cite[Section~4.1]{CastagnaPierce1994}, written
$\asub{\Delta}{A}{B}$. Figure~\ref{fig:algorithmic-subtyping} gives its
five rules in our notation.

\begin{figure}[htbp]
\begin{mathpar}
  \inferrule*[right={\ruledef{A-Top}}]{ }{\asub{\Delta}{A}{\tytop}}
  \and
  \inferrule*[right={\ruledef{A-Refl}}]{ }{\asub{\Delta}{\alpha}{\alpha}}
  \and
  \inferrule*[right={\ruledef{A-Promote}}]
    {\alpha<:A\in\Delta\\\asub{\Delta}{A}{B}}
    {\asub{\Delta}{\alpha}{B}}
  \and
  \inferrule*[right={\ruledef{A-Arrow}}]
    {\asub{\Delta}{C}{A}\\\asub{\Delta}{B}{D}}
    {\asub{\Delta}{A\to B}{C\to D}}
  \and
  \inferrule*[right={\ruledef{A-All-Top}}]
    {\asub{\Delta}{C}{A}\\\asub{\Delta,\alpha<:\tytop}{B}{D}}
    {\asub{\Delta}{\All{\alpha}{A}{B}}{\All{\alpha}{C}{D}}}
\end{mathpar}
\caption{Algorithmic subtyping, following Castagna and Pierce
  \cite[Section~4.1]{CastagnaPierce1994}. All judgments are well formed,
  and the binder in \ruleref{A-All-Top} is fresh.}
\Description{Five rules decide comparison with the top type, variable reflexivity,
left-variable promotion, arrows, and universal types. Promotion uses the
declared bound as its fixed intermediate type. Universal bodies are compared
with a variable bounded by the top type.}
\label{fig:algorithmic-subtyping}
\end{figure}

Compared with Figure~\ref{fig:subtyping}, reflexivity is restricted to
variables and general transitivity is absent. \ruleref{A-Promote} combines \ruleref{Bound}
with a \ruleref{Trans} step whose intermediate type is fixed by the context.
The premises of \ruleref{Arrow} and \ruleref{All-Top} are unchanged.
To execute the rules, first test \ruleref{A-Top} and \ruleref{A-Refl}. Otherwise, promote a
left variable or decompose matching constructors. If no rule applies,
the comparison fails.

The difference from the algorithm for full $\mathsf F_{<:}$ is confined
to the quantified body. For source $\All{\alpha}{A}{B}$ and target
$\All{\alpha}{C}{D}$, both algorithms compare $C$ with $A$. They compare
the bodies $B$ and $D$ in the respective contexts
\[
 \begin{array}{c@{\qquad\qquad}c}
   \text{full }\mathsf F_{<:} & \Ftop\\[1mm]
   \Delta,\alpha<:C & \Delta,\alpha<:\tytop.
 \end{array}
\]
In particular, \ruleref{A-All-Top} uses neither quantifier's declared bound in the
body comparison.

In the cited extended version, Castagna and Pierce prove soundness and
completeness \cite[Theorem~4.4]{CastagnaPierce1994}: for well-formed inputs,
\begin{equation}
 \asub{\Delta}{A}{B}\iff\sub{\Delta}{A}{B}.
 \label{eq:subtyping-equivalence}
\end{equation}
Their termination result \cite[Theorem~4.7 and Corollary~4.8]{CastagnaPierce1994}
makes this a decision procedure.

\paragraph{Exposure.}
Application requires an arrow head, and type application requires a
universal head. A type variable can provide either through its bound.
We use the standard bound-following operation, also described by
Laird \cite[Lemma~6.3]{Laird2023}. For well-formed $\Delta,A$, write
$\Expose{\Delta}{A}$ for the resulting type in the same scope, whose root
is $\tytop$, an arrow, or a universal:
\begin{align*}
 \Expose{\Delta}{\alpha}&=\Expose{\Delta}{\Delta(\alpha)},&
 \Expose{\Delta}{\tytop}&=\tytop,\\
 \Expose{\Delta}{(A\to B)}&=A\to B,&
 \Expose{\Delta}{(\All{\alpha}{A}{B})}&=\All{\alpha}{A}{B}.
\end{align*}
It stops as soon as the outer constructor is known, without traversing its
components. Ordered bounds make the variable case terminate. For example,
if $\Delta=\beta<:\tytop\to\tytop,\alpha<:\beta$, then
$\Expose{\Delta}{\alpha}=\tytop\to\tytop$.

Exposure is deterministic and satisfies
$\sub{\Delta}{M}{\Expose{\Delta}{M}}$. For compound targets, algorithmic
inversion gives the interfaces used at elimination:
\begin{align}
 M\subin{\Delta}C\to D
 &\iff \exists A_0,B_0.\ \Expose{\Delta}{M}=A_0\to B_0
       \notag\\[-1mm]
 &\hspace{20mm}{}\land C\subin{\Delta}A_0
                   \land B_0\subin{\Delta}D,
       \label{eq:expose-arrow}\\
 M\subin{\Delta}\All{\alpha}{C}{D}
 &\iff \exists A_0,B_0.\ \Expose{\Delta}{M}=\All{\alpha}{A_0}{B_0}
       \notag\\[-1mm]
 &\hspace{20mm}{}\land C\subin{\Delta}A_0
                   \land B_0\subin{\Delta,\alpha<:\tytop}D.
       \label{eq:expose-all}
\end{align}
Thus an exposed arrow or universal retains every comparison needed by its
elimination rule, even when the original type is a variable.
Type abstraction still requires a
way to retain all body typings when the comparison context changes from
$\alpha<:A$ to $\alpha<:\tytop$.

\section{Finite Covers under Bound Forgetting}
\label{sec:cones}

Subsumption inside a type abstraction cannot always be postponed until
after the abstraction (Section~\ref{sec:obstacles}). We construct a finite
set of alternatives for one assignable body type that covers its supertypes
after bound forgetting. Section~\ref{sec:typing-bases} uses this construction
to compute complete typing bases by recursion on terms.

\subsection{Coverage at type abstraction}
\label{sec:abstraction-coverage}

For a fresh binder $\alpha$, we use the following abbreviations throughout:
\[
 \Delta_A=\Delta,\alpha<:A,\qquad
 \Delta_\top=\Delta,\alpha<:\tytop.
\]
The surrounding expression determines the binder $\alpha$.
Consider $\tabs{\alpha}{A}{s}$ whose body has a type $B$:
\[
 \typing{\Delta_A}{\Gamma}{s}{B}.
\]
Here $B$ is one assignable type, with no assumption that it is least.
For any $R$ satisfying $B\subin{\Delta_A}R$, \ruleref{Sub} assigns $R$
to the body, and \ruleref{Type-Abs} gives the abstraction type
$\All{\alpha}{A}{R}$. If we apply \ruleref{Type-Abs} at $B$ first,
reaching that same target by subsumption requires
$B\subin{\Delta_\top}R$. As Section~\ref{sec:obstacles} shows, the known
comparison under $\Delta_A$ does not imply this comparison under
$\Delta_\top$.

To account for all such $R$, we seek a finite set $F_B$ of types $C$,
computed from $\Delta_A$ and the single type $B$, independently of $R$.
Each member must be well scoped in
$\Delta_A$, and the set must satisfy two requirements:
\begin{align}
 C\in F_B
   &\Longrightarrow B\subin{\Delta_A}C,
     \label{eq:tabs-cover-sound}\\
 B\subin{\Delta_A}R
   &\Longrightarrow
     \exists C\in F_B.\ C\subin{\Delta_\top}R.
     \label{eq:tabs-cover-complete}
\end{align}
The first requirement makes every $C$ assignable to the body by
\ruleref{Sub}. Applying \ruleref{Type-Abs} to each choice gives finitely
many types $\All{\alpha}{A}{C}$. The second requirement ensures that, for
every $R$ above $B$ in $\Delta_A$, one of these quantified types lies below
$\All{\alpha}{A}{R}$ by \ruleref{All-Top}. Any subsequent subsumption on
the quantified type remains available.
We call such a set an \emph{upper cover} of $B$ across the change from
$\Delta_A$ to $\Delta_\top$.
The requirements involve only types and bounds, so we can construct this
operation independently of the body term.

For Ghelli's example, $\Delta=\beta<:\tytop$,
$A=\beta$, and the body $\abs{x}{\alpha}{x}$ has type $B=\alpha\to\alpha$.
Take
\[
 F_B=\{\alpha\to\alpha,\alpha\to\beta\}.
\]
Both members are assignable under $\alpha<:\beta$. Applying
\ruleref{Type-Abs} to them gives the two quantified types of
Section~\ref{sec:obstacles}.
Each retained body matches its corresponding target by reflexivity
under $\alpha<:\tytop$.
Section~\ref{sec:backward} returns to the same target through an explicit
type application. There the inverse construction uses a finite family to
answer one query, rather than computing the abstraction's whole basis.

\subsection{Principal and finite covers}

Type abstraction changes only the new variable's bound. For a uniform
recursive construction, we first solve the stronger problem of forgetting
all ambient bounds. For any context $\Theta$, let $\forget{\Theta}$
replace its bounds by \tytop, preserving the variable domain and leaving
bounds inside types unchanged.

Restoring forgotten bounds preserves subtype comparisons. A use of
$\alpha<:\tytop$ can always be justified in the restored context, and \ruleref{All-Top}
extends both comparison contexts with the same fresh variable bounded by $\tytop$.
Induction on a subtype derivation therefore gives, in particular,
\begin{equation}
 \sub{\forget{\Delta_A}}{S}{T}
 \ \Longrightarrow\ \sub{\Delta_\top}{S}{T}
 \ \Longrightarrow\ \sub{\Delta_A}{S}{T}.
 \label{eq:relative-contexts}
\end{equation}
A family that provides coverage in $\forget{\Delta_A}$ consequently
provides coverage in $\Delta_\top$ as well. We can construct it by full
forgetting and then restore the outer bounds that type abstraction retains.

For a well-formed context $\Delta$, its well-scoped types form a preorder
$\typepreorder{\Delta}$ under subtyping. The \emph{principal upper and
lower cones} of a type $A$ are
\[
 \upclosure{\Delta}{A}=\{T:A\subin{\Delta}T\},
 \qquad
 \downclosure{\Delta}{A}=\{S:S\subin{\Delta}A\}.
\]
For a set $\mathcal S$ of types, extend these operations by union:
\[
 \upclosure{\Delta}{\mathcal S}
   =\bigcup_{A\in\mathcal S}\upclosure{\Delta}{A},
 \qquad
 \downclosure{\Delta}{\mathcal S}
   =\bigcup_{A\in\mathcal S}\downclosure{\Delta}{A}.
\]
Thus $\upclosure{\Delta}{\{A\}}=\upclosure{\Delta}{A}$, and similarly
for lower cones. When $\mathcal S$ is finite, these sets are \emph{finitely
generated}; the sets themselves may be infinite.
An upper cover expresses an old upper cone as a finite union of upper
cones in the weaker context.

For $\Delta=\beta<:\tytop$ and $A=\beta$ as above, the supertypes of $\alpha\to \alpha$ in $\Delta_A$ are
exactly $\tytop$, $\alpha\to \alpha$, $\alpha\to \beta$, and $\alpha\to\tytop$.
The arrow rules give
\[
 \upclosure{\Delta_A}{(\alpha\to \alpha)}
 =\upclosure{\Delta_\top}{\{\alpha\to \alpha,\alpha\to \beta\}}.
\]
In Figure~\ref{fig:cover-order}, the types form a chain under the actual
bound. In $\Delta_\top$,
$\alpha\to\alpha$ and $\alpha\to\beta$ become incomparable, while both
remain below $\alpha\to\tytop$.

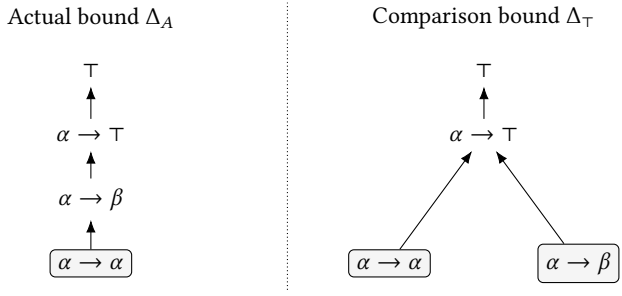
\begin{figure}[t]
\centering
\begin{tikzpicture}[
  >=Latex,
  every node/.style={font=\small},
  generator/.style={draw,rounded corners=2pt,inner sep=3pt,
                    fill=black!4}]
 \node at (-2.6,3.25) {Actual bound $\Delta_A$};
 \node at (2.6,3.25) {Comparison bound $\Delta_\top$};
 \node[generator] (a0) at (-2.6,0) {$\alpha\to\alpha$};
 \node (a1) at (-2.6,.85) {$\alpha\to\beta$};
 \node (a2) at (-2.6,1.7) {$\alpha\to\tytop$};
 \node (a3) at (-2.6,2.55) {$\tytop$};
 \draw[->] (a0)--(a1);
 \draw[->] (a1)--(a2);
 \draw[->] (a2)--(a3);
 \node[generator] (b0) at (1.35,0) {$\alpha\to\alpha$};
 \node[generator] (b1) at (3.85,0) {$\alpha\to\beta$};
 \node (b2) at (2.6,1.7) {$\alpha\to\tytop$};
 \node (b3) at (2.6,2.55) {$\tytop$};
 \draw[->] (b0)--(b2);
 \draw[->] (b1)--(b2);
 \draw[->] (b2)--(b3);
 \draw[densely dotted] (0,-.35)--(0,3.45);
\end{tikzpicture}
\caption{The principal upper cone of $\alpha\to\alpha$ under
$\alpha<:\beta$ and its finite cover under $\alpha<:\tytop$, with
$\Delta=(\beta<:\tytop)$. Arrows point toward immediate supertypes.
The outlined nodes generate the same four-element upper set in their
respective contexts: one generator on the left, two on the right.}
\Description{Two subtype-order diagrams. The actual-bound diagram is a
four-type chain beginning at alpha to alpha. The diagram with the top type
as bound is a diamond with incomparable bottom nodes alpha to alpha and
alpha to beta.
The lower nodes generating each upper set are outlined.}
\label{fig:cover-order}
\end{figure}

The left side is principal in $\Delta_A$. It has no single generator
in $\Delta_\top$, where $\alpha\to \alpha$ and $\alpha\to \beta$ have no common subtype
by the inversion argument of Section~\ref{sec:obstacles}.

Upper covers solve the obligation at type abstraction. Their construction
also needs lower covers: an arrow comparison
$A\to B<:C\to D$ requires $C<:A$ in its contravariant domain.
We therefore establish both directions together.

\begin{theorem}[Effective finite cone covers]
\label{thm:cones}
For each well-formed $\Delta,A$, one can compute finite families
$\Up{\Delta}{A}$ and $\Down{\Delta}{A}$ of types well formed in $\Delta$
such that
\begin{align}
  \upclosure{\Delta}{A}
    &=\upclosure{\forget{\Delta}}{\Up{\Delta}{A}},
      \label{eq:upper}\\
  \downclosure{\Delta}{A}
    &=\downclosure{\forget{\Delta}}{\Down{\Delta}{A}}.
      \label{eq:lower}
\end{align}
\end{theorem}

These equalities imply representative soundness. For
$u\in\Up{\Delta}{A}$, reflexivity gives
$u\in\upclosure{\forget{\Delta}}{\Up{\Delta}{A}}=\upclosure{\Delta}{A}$,
so $A\subin{\Delta}u$. Dually, every $d\in\Down{\Delta}{A}$ satisfies
$d\subin{\Delta}A$.

The singleton $\{A\}$ would suffice if both sides used $\Delta$.
The finite families need not be minimal. Restoring the outer bounds gives
coverage in the comparison context used by type abstraction.

\begin{corollary}[Restoring only the outer bounds]
\label{cor:relative-cover}
For well-formed $\Delta,A$ and types $B,T$ in the scope of $\Delta_A$,
\[
  \sub{\Delta_A}{B}{T}
  \iff \exists u\in\Up{\Delta_A}{B}.\ \sub{\Delta_\top}{u}{T}.
\]
\end{corollary}
\begin{proof}
Forward coverage holds first in $\forget{\Delta_A}$. By
\eqref{eq:relative-contexts}, restore the outer bounds to obtain the
comparison in $\Delta_\top$. Conversely, restore $\alpha$'s bound to $A$
and compose with representative soundness.
\end{proof}
Taking $F_B=\Up{\Delta_A}{B}$ now satisfies
\eqref{eq:tabs-cover-sound} by representative soundness and
\eqref{eq:tabs-cover-complete} by this corollary.

\subsection{Constructing upper and lower covers}

Given a well-formed $\Delta$ and type $A$ in its scope, the variable-source
operation returns the finite set of declared variables below $A$:
\[
  \Vars{\Delta}{A}=\{\alpha\in\domof{\Delta}:\sub{\Delta}{\alpha}{A}\}.
\]
Subtyping decides membership. With $C=\Delta(\alpha)$, the base cases are
\begin{align*}
  \Up{\Delta}{\tytop}&=\{\tytop\}, &
  \Down{\Delta}{\tytop}&=\{\tytop\},\\
  \Up{\Delta}{\alpha}&=\{\alpha\}\cup\Up{\Delta}{C}, &
  \Down{\Delta}{\alpha}&=\Vars{\Delta}{\alpha}.
\end{align*}
A variable can either remain itself or use its upper bound.
There are no compound subtypes of a variable, which accounts for the lower
variable case.

For arrows, contravariance exchanges the two recursive constructions:
\begin{align*}
  \Up{\Delta}{A\to B}
    &=\{\tytop\}\cup
      \{a\to b:a\in\Down{\Delta}{A},\ b\in\Up{\Delta}{B}\},\\
  \Down{\Delta}{A\to B}
    &=\Vars{\Delta}{A\to B}\cup
      \{a\to b:a\in\Up{\Delta}{A},\ b\in\Down{\Delta}{B}\}.
\end{align*}
The lower variable-source component cannot be omitted: a variable with a
function bound may be below an arrow without having arrow syntax.

Universal types follow the same variance:
\begin{align*}
  \Up{\Delta}{\All{\alpha}{A}{B}}
    &=\{\tytop\}\cup
      \{\All{\alpha}{a}{b}:a\in\Down{\Delta}{A},\ b\in\Up{\Delta_\top}{B}\},\\
  \Down{\Delta}{\All{\alpha}{A}{B}}
    &=\Vars{\Delta}{\All{\alpha}{A}{B}}\cup
      \{\All{\alpha}{a}{b}:a\in\Up{\Delta}{A},\ b\in\Down{\Delta_\top}{B}\}.
\end{align*}
\ruleref{All-Top} makes $\Delta_\top$ independent of the bound $a$ of a generated type,
so one recursive body cover suffices for every candidate bound.

\subsection{Termination}
\label{sec:termination}

A variable bound may be larger than the variable's syntax, so plain type
size does not justify recursion. As in Castagna and Pierce's termination
argument \cite[Section~4.1.2]{CastagnaPierce1994}, we measure types relative
to their bounds. Our universal case gives the body variable bound $\tytop$,
matching the cone recursion. Build a weight environment $\rho$ in the
order of $\Delta$. On declaring $\alpha<:A$, set
$\rho(\alpha)=1+w_\rho(A)$, where
\begin{align*}
  w_\rho(\alpha)&=\rho(\alpha),& w_\rho(\tytop)&=1,\\
  w_\rho(A\to B)&=1+w_\rho(A)+w_\rho(B),&
  w_\rho(\All{\alpha}{A}{B})&=1+w_\rho(A)+w_{\rho[\alpha\mapsto2]}(B).
\end{align*}
The weight $2$ is the weight of a variable with bound $\tytop$. Earlier weights
are available when a bound is measured, because contexts are ordered.
$\rho[\alpha\mapsto2]$ extends the weight assignment with the fresh variable $\alpha$.
Write $\tyweight{\Delta}{A}$ for the resulting type weight.

Extending the context by a fresh variable preserves the weights of old
variables. Consequently, for a type $T$ well scoped in $\Delta$,
\[
  \tyweight{\Delta,\alpha<:A}{T}=\tyweight{\Delta}{T}.
\]
Thus a variable $\alpha$ with declared bound $C$ satisfies
$\tyweight{\Delta}{C}<\tyweight{\Delta}{\alpha}$. Every cone call therefore
decreases weight: it either follows a bound, selects an arrow component,
or selects a universal component with the stipulated extension by $\alpha<:\tytop$.
Variable-source tests use the terminating subtype checker, whose recursive
calls decrease the sum of the two endpoint weights. Exposure decreases the
source weight. These measures establish termination without bounding the
number or size of generated representatives by a polynomial.

\subsection{Soundness and completeness}

\begin{proof}[Proof sketch for Theorem~\ref{thm:cones}]
Induct on the weight, proving well-formedness and representative soundness
together with both coverage equalities. A variable either remains itself
or follows its bound; lower covers also retain all qualifying variable
sources. Arrow inversion exchanges upper and lower coverage in the domain.

For a universal comparison
$\All{\alpha}{A}{B}\subin{\Delta}\All{\alpha}{C}{D}$, inversion gives
$C\subin{\Delta} A$ and $B\subin{\Delta_\top}D$.
Choose $a\in\Down{\Delta}{A}$ and $b\in\Up{\Delta_\top}{B}$ with
$C\subin{\forget{\Delta}}a$ and $b\subin{\forget{\Delta_\top}}D$.
\ruleref{All-Top} then gives
\[
  \sub{\forget{\Delta}}{\All{\alpha}{a}{b}}{\All{\alpha}{C}{D}}.
\]
The representative is sound because $a\subin{\Delta} A$ and
$B\subin{\Delta_\top}b$. The same body context $\Delta_\top$ works for every
generated bound. Lower coverage reverses the recursive choices and
retains variable sources. Appendix~\ref{app:cone-proof} gives the full induction.
\end{proof}

\section{Decidable Typechecking from Finite Bases}
\label{sec:typing-bases}

We compute typing bases by recursion on the original term.
The cone cover of Corollary~\ref{cor:relative-cover} handles bound forgetting
at type abstraction; the remaining constructors use exposure and subtype tests.

\subsection{Constructing typing bases}

\begin{theorem}[Effective finite typing bases]
\label{thm:typing-basis}
For every well-formed $\Delta;\Gamma$ and raw term $s$, a finite family
$\Basis{\Delta}{\Gamma}{s}$ satisfying \eqref{eq:typing-basis} is computable.
Every family member is an assignable type of $s$.
\end{theorem}

Figure~\ref{fig:basis-checker} collects the finite-basis checker in rule
form. The judgment $\basissynth{\Delta}{\Gamma}{s}{F}$ computes the finite
set $F=\Basis{\Delta}{\Gamma}{s}$, which may be empty;
$\basischeck{\Delta}{\Gamma}{s}{T}$ accepts the requested target exactly
when some $B\in F$ satisfies $\asub{\Delta}{B}{T}$.
The rules display the well-scoped cases with well-formed annotations.

\begin{figure}[t]
\small
\begin{mathpar}
 \inferrule*[right={\ruledef{B-Var}}]
   {\Gamma(x)=A}
   {\basissynth{\Delta}{\Gamma}{x}{\{A\}}}
 \and
 \inferrule*[right={\ruledef{B-Const}}]
   { }
   {\basissynth{\Delta}{\Gamma}{\tmconst}{\{\tytop\}}}
 \and
 \inferrule*[right={\ruledef{B-Abs}}]
   {\basissynth{\Delta}{\Gamma,x:A}{s}{F}}
   {\basissynth{\Delta}{\Gamma}{\abs{x}{A}{s}}{\{A\to B:B\in F\}}}
 \and
 \inferrule*[right={\ruledef{B-Type-Abs}}]
   {\basissynth{\Delta_A}{\Gamma}{s}{F}}
   {\basissynth{\Delta}{\Gamma}{\tabs{\alpha}{A}{s}}
      {\displaystyle\bigcup_{B\in F}
        \{\All{\alpha}{A}{C}:C\in\Up{\Delta_A}{B}\}}}
 \and
 \inferrule*[right={\ruledef{B-App}}]
   {\basissynth{\Delta}{\Gamma}{f}{F}\\
    \basissynth{\Delta}{\Gamma}{u}{G}}
   {\basissynth{\Delta}{\Gamma}{f\,u}
      {\{B:C\in F,\ \Expose{\Delta}{C}=A\to B,\
           D\in G,\ \asub{\Delta}{D}{A}\}}}
 \and
 \inferrule*[right={\ruledef{B-Type-App}}]
   {\basissynth{\Delta}{\Gamma}{f}{F}}
   {\basissynth{\Delta}{\Gamma}{f\{R\}}
      {\{B[R/\alpha]:C\in F,\
         \Expose{\Delta}{C}=\All{\alpha}{A}{B},\ \asub{\Delta}{R}{A}\}}}
 \and
 \inferrule*[right={\ruledef{B-Check}}]
   {\basissynth{\Delta}{\Gamma}{s}{F}\\
    \exists B\in F.\ \asub{\Delta}{B}{T}}
   {\basischeck{\Delta}{\Gamma}{s}{T}}
\end{mathpar}
\caption{Finite-basis synthesis and checking.}
\Description{Six synthesis rules cover variables, the term constant Top, term
and type abstraction, application and type application. Type abstraction
takes upper covers of its body candidates. Elimination filters finite
candidate families using head exposure and subtype tests. The final
checking rule accepts when a synthesized candidate is below the target.}
\label{fig:basis-checker}
\end{figure}

Each recursive basis supplies assignable candidates for the enclosing
typing rule. Ordinary abstraction preserves coverage through covariance
in the result type. At eliminations, exposure allows a candidate with a
variable head to supply an arrow or universal. The tests in
\ruleref{B-App} and \ruleref{B-Type-App} establish the premises needed to use
that exposed type. For ordinary application, one argument candidate below
the exposed domain suffices by subsumption. Type application checks the
explicit type argument against the exposed bound.

\ruleref{B-Type-Abs} is where the cone cover is needed. A body representative
covers its supertypes under the actual bound in $\Delta_A$.
Directly quantifying that representative can lose coverage, since
\ruleref{All-Top} compares bodies in $\Delta_\top$.
The rule first takes its upper cover and then quantifies each retained
body under the written bound $A$. Representative soundness makes each
candidate assignable; coverage after forgetting preserves every typing
of the abstraction. This discharges the obligation from
Section~\ref{sec:abstraction-coverage}.

A missing variable, malformed annotation or malformed type argument
yields the empty basis, as does an elimination with no surviving candidate.

\subsection{Soundness and completeness}

\begin{proof}[Proof sketch for Theorem~\ref{thm:typing-basis}]
Induct on $s$. Every recursive call is on a proper subterm, and every
auxiliary construction terminates by Section~\ref{sec:termination}.
Typing inversion exposes the rule for that term constructor followed by a
possible final subsumption.

At type abstraction, let inversion give a body type $R$ and a comparison
$\All{\alpha}{A}{R}\subin{\Delta} T$. Choose a body representative $B$ with
$B\subin{\Delta_A}R$. Corollary~\ref{cor:relative-cover} gives
$C\in\Up{\Delta_A}{B}$ with
\[
  \sub{\Delta,\alpha<:\tytop}{C}{R}.
\]
Hence \ruleref{All-Top} gives $\All{\alpha}{A}{C}<:_\Delta\All{\alpha}{A}{R}$ and then
$\All{\alpha}{A}{C}\subin{\Delta} T$. Soundness uses $B\subin{\Delta_A}C$ to type
the body before applying \ruleref{Type-Abs}. Cone coverage thus accounts for
the subsumption before type abstraction in Section~\ref{sec:obstacles}.

At a type application, inversion gives $f:\All{\alpha}{U}{V}$, $R\subin{\Delta} U$,
and $V[R/\alpha]\subin{\Delta} T$. By \eqref{eq:expose-all}, a function representative exposes as
$\All{\alpha}{A}{B}$ with $U\subin{\Delta} A$ and
$B\subin{\Delta,\alpha<:\tytop}V$. The bound test $R\subin{\Delta} A$ succeeds.
Preservation of subtyping by substitution from $\Delta_\top$ gives
$B[R/\alpha]\subin{\Delta} V[R/\alpha]$, so the retained type covers $T$.
Soundness follows from exposure, \ruleref{Type-App}, and the successful bound test.
All these arguments retain the final subsumption supplied by inversion.
Appendix~\ref{app:typing-basis-proof} gives the full induction over term constructors.
\end{proof}

\begin{corollary}[Decidability of the original judgment]
\label{cor:decidable}
Typechecking for arbitrary explicitly annotated \Ftop{} terms is decidable.
There is a total procedure that accepts a raw input $(\Delta,\Gamma,s,T)$
exactly when $\typing{\Delta}{\Gamma}{s}{T}$ holds.
\end{corollary}
\begin{proof}
Reject malformed environments and targets. Otherwise compute the finite
basis and ask whether a member is below $T$. Representative soundness and
complete coverage prove the equivalence. Malformed inputs admit no typing
derivation, so the wrapper needs no well-formedness premise on its caller.
\end{proof}

\subsection{The cost of finite bases}
\label{sec:basis-cost}

The finite-basis theorem establishes a terminating construction but gives
no polynomial bound on its cost. Shared representations can compress some
very large finite families (Appendix~\ref{sec:large-families}). Such
compression alone does not bound the cost of general construction or
queries.

\section{Backward Checking by Inverse Substitution}
\label{sec:backward}

The finite-basis procedure summarizes all assignable types before comparing
with a target. To decide a single judgment, we instead propagate acceptable
targets through the term. Application requires a domain for the argument
check; type application requires body targets before substitution.
We introduce rectangularity to address the first problem, 
while finite inverse coverage addresses the second.

\subsection{A canonical interface for elimination}
\label{sec:canonical-domains}

Exposing a \emph{canonical type} must recognize every available elimination
shape and supply a domain or bound that preserves all result choices,
including after substitution.

The rules below define $\canrel{\Delta}{\Gamma}{s}{C}$ on well-formed
contexts and annotations, using the original typing judgment for argument
checks. Section~\ref{sec:backward-equations} implements this relation.
\begin{mathparpagebreakable}[\clubpenalty=0\widowpenalty=0\interlinepenalty=0]
 \inferrule{\Gamma(x)=A}
   {\canrel{\Delta}{\Gamma}{x}{A}}
 \and
 \inferrule{ }
   {\canrel{\Delta}{\Gamma}{\tmconst}{\tytop}}
 \\
 \inferrule{\canrel{\Delta}{\Gamma,x:A}{s}{B}}
   {\canrel{\Delta}{\Gamma}{\abs{x}{A}{s}}{A\to B}}
 \and
 \inferrule{\canrel{\Delta,\alpha<:A}{\Gamma}{s}{B}}
   {\canrel{\Delta}{\Gamma}{\tabs{\alpha}{A}{s}}{\All{\alpha}{A}{B}}}
 \\
 \inferrule{\canrel{\Delta}{\Gamma}{f}{C}\\\\
   \Expose{\Delta}{C}=A\to B\\\typing{\Delta}{\Gamma}{u}{A}}
   {\canrel{\Delta}{\Gamma}{f\,u}{B}}
 \and
 \inferrule{\canrel{\Delta}{\Gamma}{f}{C}\\\\
   \Expose{\Delta}{C}=\All{\alpha}{A}{B}\\R\subin{\Delta}A}
   {\canrel{\Delta}{\Gamma}{f\{R\}}{B[R/\alpha]}}
\end{mathparpagebreakable}

The proof tracks the pair formed by a term's complete typing family and
its canonical type. These semantic families may be infinite and are used
only in the proof.

A substitution $\sigma:\Delta\to\Theta$ is \emph{typed} when its images
are well formed and each source variable $\alpha$ satisfies
$\sigma(\alpha)\subin{\Theta}\Delta(\alpha)[\sigma]$. Such substitutions
preserve subtyping.
For a type set $\mathcal S$, write
$\mathcal S[\sigma]=\{B[\sigma]:B\in\mathcal S\}$.

\begin{definition}[Structured families]
\label{def:structured-families}
A pair $(\mathcal S,C)$ is \emph{structured} in $\Delta$ if it is
generated by the following clauses. All types are well formed in their
indicated contexts.
\begin{enumerate}
 \item \emph{Singleton.} $(\{C\},C)$ is structured in $\Delta$.
 \item \emph{Constructors.} From $(\mathcal S,B)$ structured in $\Delta$,
   form $(\{A\to V:V\in\mathcal S\},A\to B)$ in $\Delta$.
   From $(\mathcal S,B)$ structured in $\Delta,\alpha<:A$, form
   $(\{\All{\alpha}{A}{V}:V\in\mathcal S\},\All{\alpha}{A}{B})$ in $\Delta$.
 \item \emph{Closure replacement.} Replace $\mathcal S$ by any family
   containing $C$ with the same upward closure in $\Delta$.
 \item \emph{Typed image.} Along a typed $\sigma:\Delta\to\Theta$, form
   $(\mathcal S[\sigma],C[\sigma])$ in $\Theta$.
\end{enumerate}
\end{definition}

Singletons account for variables and the constant; constructors account
for introductions. Closure replacement accounts for subsumption, and
typed images account for type application. Every structured family
contains its distinguished member, so it may be replaced by its upward
closure. The next lemma shows that elimination also returns a structured
family, allowing the invariant to be used again.

\begin{lemma}[Rectangularity under typed substitution]
\label{lem:structured-rectangularity}
Let $(\mathcal S,C)$ be structured in $\Delta$ and
$\sigma:\Delta\to\Theta$ be typed. If
$\Expose{\Theta}{C[\sigma]}=A_0\to B_0$, there is a structured pair
$(\mathcal R,B_0)$ in $\Theta$ such that, for all well-scoped $U,V$,
\[
 U\to V\in\upclosure{\Theta}{\mathcal S[\sigma]}
 \iff U\subin{\Theta}A_0\ \land\ V\in\mathcal R.
\]
If $\Expose{\Theta}{C[\sigma]}=\All{\alpha}{A_0}{B_0}$, there is a
structured pair $(\mathcal R,B_0)$ in $\Theta,\alpha<:A_0$ such that
\[
 \All{\alpha}{U}{V}\in\upclosure{\Theta}{\mathcal S[\sigma]}
 \iff U\subin{\Theta}A_0\ \land\ V\in\mathcal R,
\]
where $V$ has the extended variable scope. Conversely, an arrow in
$\upclosure{\Theta}{\mathcal S[\sigma]}$ forces $C[\sigma]$ to expose an
arrow, and a universal forces it to expose a universal.
\end{lemma}
\begin{proof}[Proof sketch]
Induct on the structured-family clauses, generalizing over $\Theta$ and
$\sigma$. Singleton families use exposure and subtype inversion.
Constructor families keep the domain independent of the body choices.
In the universal case, \ruleref{All-Top} takes the body's upward closure
under bound $\tytop$, but the body family must remain structured under
the actual bound $A$. For a family $\mathcal S$ in $\Delta_A$, bound
restoration and transitivity give
\begin{equation}
 \upclosure{\Delta_A}{(\upclosure{\Delta_\top}{\mathcal S})}
 =\upclosure{\Delta_A}{\mathcal S}.
 \label{eq:structured-top-closure}
\end{equation}
Thus closure replacement preserves the structured pair under the actual
bound, also after substitution. Subtype preservation transports closure
replacement through $\sigma$; the typed-image case uses composition of
typed substitutions. Appendix~\ref{app:rectangularity-proof} gives the
full induction.
\end{proof}

\begin{corollary}[Canonical elimination interface]
\label{cor:canonical-rectangularity}
In well-formed contexts $\Delta;\Gamma$, the canonical relation satisfies:
\begin{enumerate}
 \item \emph{Typability and assignability.} A term $s$ is typable exactly
   when there is a $C$ with $\canrel{\Delta}{\Gamma}{s}{C}$. This $C$ is
   unique and $\typing{\Delta}{\Gamma}{s}{C}$.
 \item \emph{Shape completeness.} For such a $C$, $s$ has an assignable
   arrow type exactly when $\Expose{\Delta}{C}$ is an arrow, and an
   assignable universal type exactly when it is a universal.
 \item \emph{Rectangularity.} If $\Expose{\Delta}{C}=A\to B$ or
   $\Expose{\Delta}{C}=\All{\alpha}{A}{B}$, respectively, then
   \begin{align}
    \typing{\Delta}{\Gamma}{s}{U\to V}
    &\iff U\subin{\Delta}A\ \land\
          \typing{\Delta}{\Gamma}{s}{A\to V},
    \label{eq:domain-dominance}\\
    \typing{\Delta}{\Gamma}{s}{\All{\alpha}{U}{V}}
    &\iff U\subin{\Delta}A\ \land\
          \typing{\Delta}{\Gamma}{s}{\All{\alpha}{A}{V}}.
    \label{eq:bound-dominance}
   \end{align}
   Components range over their respective variable scopes, with universal
   bodies in the extended scope.
 \item \emph{Substitution stability.} The pair
   $(\{T:\typing{\Delta}{\Gamma}{s}{T}\},C)$ is structured. Its image
   under every typed substitution remains structured and satisfies
   Lemma~\ref{lem:structured-rectangularity}, including the structured
   result and body families needed for subsequent eliminations.
\end{enumerate}
\end{corollary}
\begin{proof}[Proof sketch]
Induction on typing constructs a canonical derivation and its structured
typing family together. Variables and introductions use the corresponding
structured-family clauses and upward closure. At application,
Lemma~\ref{lem:structured-rectangularity} supplies an exposed arrow,
permits subsuming the argument to its domain, and returns a structured
result family.

At type application, it supplies a structured body pair $(\mathcal R,B)$
under the actual bound $A$. The bound test $R\subin{\Delta}A$ makes
$[R/\alpha]:(\Delta,\alpha<:A)\to\Delta$ typed. The complete result family is
\[
 \{T:\typing{\Delta}{\Gamma}{f\{R\}}{T}\}
   =\upclosure{\Delta}{\mathcal R[R/\alpha]},
\]
structured with member $B[R/\alpha]$ by typed image and closure replacement.
This preserves the invariant when the result is eliminated again.
Induction on canonical derivations gives assignability and uniqueness;
Lemma~\ref{lem:canonical} gives the detailed typing induction.
Finally, apply Lemma~\ref{lem:structured-rectangularity} with identity
substitution to the upward-closed typing family. Its exposed heads give
shape completeness, and reflexivity at $A$ gives the two rectangularity
equations with the same $V$.
\end{proof}

The same $V$ occurs on both sides of each equation: possible domains are
exactly the types below $A$, independently of the possible results. This
is the \emph{rectangularity} of the complete typing family. An arbitrary
finite family of arrows can still correlate domains and results.

For an example of successive eliminations, let
$i=\tabs{\alpha}{\tytop}{\abs{x}{\alpha}{x}}$,
$R=\tytop\to\tytop$, and $u=\abs{z}{\tytop}{z}$. The canonical rules give
\[
 \begin{aligned}
 &\canrel{\varnothing}{\varnothing}{i\{R\}}{R\to R},\\
 &\canrel{\varnothing}{\varnothing}{i\{R\}\,u}{R}.
 \end{aligned}
\]
The result can now be applied to $\tmconst$. Before substitution, the
corresponding result type was the variable $\alpha$; its arrow head comes
from $R$. Typed images and structured result families allow the same
shape and domain properties to justify this next elimination.

\begin{corollary}[Canonical elimination equations]
\label{cor:backward-elimination}
With $A$ the exposed canonical domain or bound, respectively,
\begin{align}
 \typing{\Delta}{\Gamma}{f\,u}{T}
 &\iff \typing{\Delta}{\Gamma}{u}{A}\ \land\
         \typing{\Delta}{\Gamma}{f}{A\to T},\label{eq:back-app}\\
 \typing{\Delta}{\Gamma}{f\{R\}}{T}
 &\iff R\subin{\Delta}A\ \land\
   \exists V.\ \typing{\Delta}{\Gamma}{f}{\All{\alpha}{A}{V}}
                    \land V[R/\alpha]\subin{\Delta}T.\label{eq:back-tapp}
\end{align}
\end{corollary}
\begin{proof}
Inversion supplies $f:U\to V$, $u:U$, and $V<:T$.
By \eqref{eq:domain-dominance}, $U<:A$ and $f:A\to V$, hence $u:A$ and
$f:A\to T$. Type application uses \eqref{eq:bound-dominance} with the
same body $V$. Both converses follow by elimination and subsumption.
\end{proof}

\subsection{Finite inverse substitution}
\label{sec:inverse}
\label{sec:generators}

Equation~\eqref{eq:back-app} determines the next target for $f$.
Equation~\eqref{eq:back-tapp} still quantifies over an unknown $V$, constrained
by its substituted form $V[R/\alpha]$. We need a finite description of acceptable
bodies before that substitution. Subsequent checking steps must be able to
transform this whole description again.

Consider a type application of Ghelli's example. Fix
$\Delta=(\beta<:\tytop)$ and an empty term context, and write
\begin{equation}
 \begin{aligned}
 f&=\tabs{\gamma}{\beta}{\tabs{\alpha}{\gamma}{\abs{x}{\alpha}{x}}},
 &s&=f\{\beta\},
 &T&=\All{\alpha}{\beta}{\alpha\to\beta}.
 \end{aligned}
 \label{eq:running-backward-query}
\end{equation}
The reduct of $s$ is the term from Section~\ref{sec:obstacles}.
Its canonical type is $\All{\alpha}{\beta}{\alpha\to\alpha}$, which is not
below $T$. One body target before substitution is
\begin{equation}
 W=\All{\alpha}{\gamma}{\alpha\to \gamma},\qquad W[\beta/\gamma]=T.
 \label{eq:running-backward-witness}
\end{equation}
Assigning $\All{\gamma}{\beta}{W}$ to $f$ suffices to check $s$ at $T$.
The inverse construction computes a finite family covering all acceptable
bodies, including this choice. Section~\ref{sec:backward-example} follows
it through the complete check.

For the general construction, fix well-formed contexts $\Delta,\Theta$,
a typed substitution $\sigma:\Delta\to\Theta$, and a target $T$ well formed
in $\Theta$. We seek finite families $\invlo{\sigma}{T}$ and
$\invup{\sigma}{T}$ in the source context $\Delta$. Writing
$\genlang{F}$ for the set of types represented by a family $F$, their
specifications are, for every $B$ well formed in $\Delta$,
\begin{align}
 B[\sigma]\subin{\Theta}T
 &\iff\exists V\in\genlang{\invlo{\sigma}{T}}.\ B\subin{\Delta}V,
       \label{eq:lower-inverse}\\
 T\subin{\Theta}B[\sigma]
 &\iff\exists V\in\genlang{\invup{\sigma}{T}}.\ V\subin{\Delta}B.
       \label{eq:upper-inverse}
\end{align}
The lower family generates exactly the source types whose substituted
forms lie below $T$. An acceptable body may be any subtype of a retained
representative. The upper family has the dual specification and is needed
when recursion enters a contravariant domain.

To propagate several acceptable targets and share their constructor
structure, we represent finite families by the grammar
\[
 F,G::=\varnothing\mid\{A\}\mid F\cup G
       \mid\garr{F}{G}\mid\gall{F}{G}.
\]
The denotation $\mathcal L$ respects empty families, singletons and choice:
$\genlang{\varnothing}=\varnothing$, $\genlang{\{A\}}=\{A\}$, and
$\genlang{F\cup G}=\genlang{F}\cup\genlang{G}$.
The arrow constructor takes two families in $\Delta$ and returns a family
in $\Delta$. The universal constructor takes a domain family in $\Delta$
and a body family in the extended variable scope, returning a family in
$\Delta$. Choose a fresh $\alpha$ bound in the body family $G$ but not in $F$.
These scope requirements concern variable declarations only. The product
denotations are
\begin{align*}
 \genlang{\garr{F}{G}}
   &=\{A\to B:A\in\genlang{F},\ B\in\genlang{G}\},\\
 \genlang{\gall{F}{G}}
   &=\{\All{\alpha}{A}{B}:A\in\genlang{F},\ B\in\genlang{G}\}.
\end{align*}
A finite set written as a family abbreviates a choice of singletons.
In a family argument position, a type $T$ stands for $\{T\}$.
Membership in a family means membership in its denotation.
Choices retain their own children. Replacing $\{A\to B,C\to D\}$ by
$\garr{\{A,C\}}{\{B,D\}}$ can introduce the additional types $A\to D$
and $C\to B$. A product permits independent choices on its two sides,
including when the same family occurs in both positions.

For fixed well-formed $\Delta;\Gamma$, the backward judgment checks a
term $s$ against a finite target representation $F$ in $\Delta$. It holds when
some represented target is assignable to the term:
\begin{equation}
 \checkfam{\Delta}{\Gamma}{s}{F}
 \iff \exists T\in\genlang{F}.\ \typing{\Delta}{\Gamma}{s}{T}.
 \label{eq:target-family-check}
\end{equation}
Here $F$ describes acceptable targets;
$\Basis{\Delta}{\Gamma}{s}$ summarizes all assignable types.
The initial query has $F=\{T\}$. Later queries may contain several
alternatives even though the original request has only one target.

The signed subtype predicates take a finite family $F$ and a type $Q$,
both in the displayed comparison context $\Theta$, and express existence
of a subtype or supertype witness:
\[
 \qry{+}{\Theta}{F}{Q}\iff\exists A\in\genlang{F}.\ A\subin{\Theta}Q,
 \qquad
 \qry{-}{\Theta}{F}{Q}\iff\exists A\in\genlang{F}.\ Q\subin{\Theta}A.
\]
The sign records which side of subtyping contains the family: the subtype
for $+$ and the supertype for $-$.
The single-target specifications extend to finite target families as follows.

\begin{theorem}[Finite inverse families]
\label{thm:inverse}
For well-formed $\Delta,\Theta$, a typed $\sigma:\Delta\to\Theta$, and a
finite family $F$ well scoped in $\Theta$, finite families
$\invfam{p}{\sigma}{F}$ are computable in $\Delta$ for $p\in\{-,+\}$.
For every $Q$ well formed in the source context $\Delta$,
\begin{equation}
 \qry{p}{\Delta}{\invfam{p}{\sigma}{F}}{Q}
 \iff \qry{p}{\Theta}{F}{\tysubst{Q}{\sigma}}.
 \label{eq:inverse-family}
\end{equation}
Every generated representative is well scoped and substitutes below
(minus) or above (plus) some member of $F$.
\end{theorem}
\begin{proof}[Construction and proof]
Construct both signs simultaneously by induction on the singleton target's
weight from Section~\ref{sec:termination}, ranging over source contexts
and typed substitutions. A source variable $\alpha$ qualifies for a variable
guard precisely when
$\sigma(\alpha)\subin{\Theta}T$ for the lower inverse, or
$T\subin{\Theta}\sigma(\alpha)$ for the upper inverse. These finitely many
decidable tests cover all variable candidates.

For an arrow candidate and target, inversion gives
\[
 (B_1\to B_2)[\sigma]\subin{\Theta}C\to D
 \iff C\subin{\Theta}B_1[\sigma]\ \land\
       B_2[\sigma]\subin{\Theta}D.
\]
Upper coverage of $C$ supplies $a$ with $a\subin{\Delta}B_1$.
Lower coverage of $D$ supplies $b$ with $B_2\subin{\Delta}b$.
Their arrow $a\to b$ covers the candidate $B_1\to B_2$ from above.
The corresponding equation for the reference construction is
\begin{equation}
 \begin{split}
 \invlo{\sigma}{C\to D}
 ={}&\{\alpha\in\domof{\Delta}:
             \sigma(\alpha)\subin{\Theta}C\to D\}\\
 &{}\cup\garr{\invup{\sigma}{C}}{\invlo{\sigma}{D}}.
 \end{split}
 \label{eq:inverse-arrow-construction}
\end{equation}
Representative soundness gives $(a\to b)[\sigma]\subin{\Theta}C\to D$.
The reversed comparison exchanges the two recursive signs. Universal
bodies use the lifted substitution under a binder with bound $\tytop$ on
both sides, independently of the chosen bound. A positive target variable
follows its bound, decreasing the weight. The cases for $\tytop$ follow by
subtype inversion.

Soundness follows from the guards and constructor rules. Inversion gives
recursive witnesses for completeness. Extend the construction to choice by
union and to independent products by combining their child witnesses with
the appropriate signs. An empty target family has empty inverses in both
directions. In particular, a product with an empty child contributes no
representatives. Subtype
preservation by typed substitution and representative soundness give the
reverse implications.
Appendix~\ref{app:inverse-construction} gives all construction equations
and the simultaneous soundness and coverage induction.
\end{proof}

Taking $\sigma=\idsubst:\forget{\Delta}\to\Delta$ recovers the cone
specifications \eqref{eq:lower} and~\eqref{eq:upper}, together with member
soundness. This identity is typed because its source bounds are all $\tytop$.
The syntax of a type is unchanged, but its two comparison contexts differ.
The specifications coincide, without requiring identical representative lists.

\paragraph{The comparison bound \texorpdfstring{$\tytop$}{top}.}
\label{sec:single-hole}

The checker needs $\sigma=[R/\alpha]:(\Delta,\alpha<:\tytop)\to\Delta$, lifted under
local binders. Combining its lower inverse with \eqref{eq:back-tapp} gives
\begin{equation}
 \typing{\Delta}{\Gamma}{f\{R\}}{T}
 \iff R\subin{\Delta}A\ \land\
 \exists W\in\genlang{\invlo{[R/\alpha]}{T}}.\
       \typing{\Delta}{\Gamma}{f}{\All{\alpha}{A}{W}}.
 \label{eq:back-inverse}
\end{equation}
Inverse completeness gives $V\subin{\Delta,\alpha<:\tytop}W$.
\ruleref{All-Top} lifts this comparison to quantified types, assigning $f$ the type
$\All{\alpha}{A}{W}$. Conversely, inverse soundness gives $W[R/\alpha]<:T$, so
\ruleref{Type-App} followed by subsumption proves the required typing.
The source bound $\alpha<:\tytop$ is necessary for this argument: a comparison
derived only under the actual bound $A$ would not justify the \ruleref{All-Top} step.

\subsection{Checking equations}
\label{sec:backward-equations}

The partial synthesis function $\canon{\Delta}{\Gamma}{s}$ implements the
canonical rules of Section~\ref{sec:canonical-domains}. It evaluates
canonical premises recursively and subtype premises by the subtyping
decider. At application, it implements the declarative guard
$\typing{\Delta}{\Gamma}{u}{A}$ by the family check
$\checkfam{\Delta}{\Gamma}{u}{\{A\}}$ defined below. Failed lookup,
ill-formed annotations or type arguments, wrong exposed heads, and failed
guards reject. Theorem~\ref{thm:backward} proves synthesis and checking
correct together.

Checking $\abs{x}{A}{s}$ against $C\to D$ requires
$C\subin{\Delta}A$ and a body check at $D$. For a finite target family
$F$, filtered projections $\mathsf P^\to$ and $\mathsf P^\forall$ retain
the complete body families of qualifying branches. These families are the
targets passed to the abstraction's body.
With $A$ well formed in $\Delta$, the arrow projection stays in $\Delta$;
the universal projection returns a family in $\Delta,\alpha<:A$, renaming
target binders to the abstraction's fresh $\alpha$:
\begin{align*}
 \genlang{\bodysel{\Delta}{\to}{F}{A}}
   &=\{D:\exists C.\ C\to D\in\genlang{F}
                        \land C\subin{\Delta}A\},\\
 \genlang{\bodysel{\Delta}{\forall}{F}{A}}
   &=\{D:\exists C.\ \All{\alpha}{C}{D}\in\genlang{F}
                        \land C\subin{\Delta}A\}.
\end{align*}
Both domain tests use the displayed outer context $\Delta$.
The projections distribute over choice and inspect singleton types. For a
matching product, retain its body child exactly when the positive query
on its domain child and $A$ succeeds. Domain and result choices therefore
remain within the same branch.

Each checking call first runs canonical synthesis
and rejects on failure or an empty target family. Success establishes typability, supplies elimination domains,
and discharges the guards $u:A$ and $R\subin{\Delta}A$.
If $\tytop\in\genlang{F}$, accept. Otherwise use the equations below.
In the final two lines, $A$ is the domain or bound exposed from
$\canon{\Delta}{\Gamma}{f}$. The last line uses the inverse
of \eqref{eq:back-inverse}, with source context $\Delta_\top$:
\begin{align*}
 \checkfam{\Delta}{\Gamma}{x}{F}
   &\iff\qry{-}{\Delta}{F}{\Gamma(x)},\\
 \checkfam{\Delta}{\Gamma}{\abs{x}{A}{s}}{F}
   &\iff\checkfam{\Delta}{\Gamma,x:A}{s}{\bodysel{\Delta}{\to}{F}{A}},\\
 \checkfam{\Delta}{\Gamma}{\tabs{\alpha}{A}{s}}{F}
   &\iff\checkfam{\Delta_A}{\Gamma}{s}{\bodysel{\Delta}{\forall}{F}{A}},\\
 \checkfam{\Delta}{\Gamma}{f\,u}{F}
   &\iff\checkfam{\Delta}{\Gamma}{f}{\garr{\{A\}}{F}},\\
 \checkfam{\Delta}{\Gamma}{f\{R\}}{F}
   &\iff\checkfam{\Delta}{\Gamma}{f}{\gall{\{A\}}{\invlo{[R/\alpha]}{F}}}.
\end{align*}
The term $\tmconst$ is rejected when $\tytop\notin\genlang{F}$.
The last two lines follow from the elimination and inverse equations.
Appendix~\ref{app:cost} implements these equations by producing candidates
along a whole checking spine and saving its elimination domains and scoped
obligations. Discharging these obligations once supplies the canonical
premises for subsequent checks on that spine, whose changing target families
do not require synthesis to be repeated.

\begin{theorem}[Backward checking correctness]
\label{thm:backward}
Canonical synthesis and family checking terminate. For well-formed contexts,
synthesis returns $C$ exactly when $\canrel{\Delta}{\Gamma}{s}{C}$ holds,
and checking any well-scoped finite family satisfies
\eqref{eq:target-family-check}.
With malformed inputs rejected, the procedure accepts a raw input
$(\Delta,\Gamma,s,T)$ exactly when $\typing{\Delta}{\Gamma}{s}{T}$ holds.
It consequently agrees with the finite-basis checker.
\end{theorem}
\begin{proof}
Use strong induction on term size, generalizing over well-formed contexts
and finite target families. At each size, first prove that synthesis
terminates and returns exactly the types related by $\vdash_{\mathrm c}$;
then prove termination and correctness of family checking at that size.

Synthesis recurses only on proper subterms. By the synthesis induction
hypothesis, its recursive results coincide with the canonical premises.
For an application $f\,u$, the checking induction hypothesis for the smaller
term $u$ makes the test at $\{A\}$ equivalent to the declarative guard $u:A$.
Subtype decidability supplies the type-application guard; lookup, formation
and exposure tests implement the remaining premises. Thus synthesis
succeeds exactly when the canonical relation holds. Lemma~\ref{lem:canonical}
then gives an assignable result on success and excludes every typing on failure.

Family checking at the current size uses this already established synthesis
result. An empty family rejects. When $\tytop\in\genlang{F}$, checking succeeds
exactly when synthesis does. All remaining recursive checks are on a proper
body or function child, so the checking induction hypothesis applies to
their transformed target families. Projections and signed queries terminate,
and finite inverse operations terminate independently by Theorem~\ref{thm:inverse}.

Variables and introductions follow inversion and branch selection.
For type abstraction, restore body comparisons from $\Delta_\top$ to $\Delta_A$.
Conversely, a selected body type can be used in \ruleref{Type-Abs}. Reflexivity
in $\Delta_\top$ and the successful domain test justify final subsumption.
Equations~\eqref{eq:back-app} and~\eqref{eq:back-inverse} prove elimination.
The wrapper rejects malformed environments and targets. Synthesis rejects
invalid annotations and type arguments; typing formation excludes a
derivation in each case.
\end{proof}

Only inverse exactness enters the type-application argument. Any total
inverse construction satisfying that interface gives the same checking
answers: each computes the target-family judgment
\eqref{eq:target-family-check}. The general construction and the single-hole
specialization of Section~\ref{app:single-hole-refinement} are instances of
this result.

\subsection{Beyond the canonical upper cone}
\label{sec:backward-example}

For the query in \eqref{eq:running-backward-query}, the canonical subtype
test fails, but the target family has an accepting branch. The canonical
bound of $f$ is $\beta$, so its type-argument guard succeeds. The lower
inverse contains the body $W$ of \eqref{eq:running-backward-witness}.
The positive guard on the quantified domain and the negative guard on
the arrow result both test $\beta<:\beta$ and permit the fresh $\gamma$;
the arrow domain retains $\alpha$. Equations~\eqref{eq:back-inverse} and the
body projections then give:
\begin{center}
\begin{tabular}{lll}
\toprule
Current term & Selected target & Successful test or choice\\
\midrule
$f\{\beta\}$ & $T$ & $\beta<:\beta$; inverse member $W$\\
$f$ & $\All{\gamma}{\beta}{W}$ & $\beta<:\beta$; retain body $W$\\
$\tabs{\alpha}{\gamma}{\abs{x}{\alpha}{x}}$ & $W$ & $\gamma<:\gamma$; retain $\alpha\to\gamma$\\
$\abs{x}{\alpha}{x}$ & $\alpha\to\gamma$ & $\alpha<:\alpha$; retain $\gamma$\\
$x$ & $\gamma$ & $\Gamma(x)=\alpha$, $\alpha<:\gamma$\\
\bottomrule
\end{tabular}
\end{center}
The recursive checks enter the actual contexts $\gamma<:\beta$ and then
$\alpha<:\gamma$, where the last variable test succeeds. The complete
inverse and body families retain their other branches.

\section{Symbolic Implementation}
\label{sec:symbolic-implementation}

The backward checker consumes finite families through signed queries,
body selection and inverse substitution. The specification does not
require enumerating their denotations. We specialize the inverse to the
substitutions generated by type application, then share the resulting
family structure and its comparison queries.

\subsection{Single-hole inverse specialization}
\label{app:single-hole-refinement}

For a type-application substitution, old source variables need no separate
guards: retaining the original targets already covers them.

Fix a comparison context $\Theta$, a replacement $R$ well formed there,
and a fresh hole $\alpha$. The typed substitution
$[R/\alpha]:(\Theta,\alpha<:\tytop)\to\Theta$ maps the new source scope back to the
target scope. Input families are in $\Theta$, and inverse families are
in $\Theta,\alpha<:\tytop$.
Let $\weaksubst:\Theta\hookrightarrow\Theta,\alpha<:\tytop$ be the scope
inclusion, fixing each old variable.
For every type $T$ in $\Theta$,
\[
 (T[\weaksubst])[R/\alpha]=T.
\]
Substitution therefore retracts this embedding. The recursive rules below
retain every old target member in the extended scope. These copies cover
the old-variable alternatives tested by the general inverse.

For example, let $\beta$ be an old variable and $T$ a member of the target
family. In the lower inverse, the general construction adds $\beta$ when
$\beta\subin{\Theta}T$. Weakening preserves that comparison:
\[
 \beta\subin{\Theta}T
 \quad\Longrightarrow\quad
 \beta\subin{\Theta,\alpha<:\tytop}T.
\]
The retained $T$ covers $\beta$ and, by transitivity, every
source type below $\beta$. Omitting the separate $\beta$ alternative
therefore preserves all lower-inverse queries. Dually, if
$T\subin{\Theta}\beta$, the retained $T$ lies below $\beta$
and covers its upper-inverse alternative. For the fresh hole $\alpha$,
substitution instead produces $R$, so its eligibility is decided by
comparing $R$ with the target family.

Write $\invshort{p}{F}$ for the specialized inverse below. It has the
coverage specification of $\invfam{p}{[R/\alpha]}{F}$, but need not return the
same representatives. For a family $F$ in the fixed target context $\Theta$, the
fresh-hole guard returns either the singleton $\{\alpha\}$ or the empty family
in the extended source scope:
\[
 \invgate{p}{F}=\begin{cases}
   \{\alpha\},&\qry{p}{\Theta}{F}{R},\\
   \varnothing,&\text{otherwise}.
 \end{cases}
\]
With this scope convention, the pruned rules are
\begin{align*}
 \invshort{-}{\varnothing}&=\invshort{+}{\varnothing}=\varnothing,\\
 \invshort{-}{\tytop}&=\{\tytop\},&
 \invshort{+}{\tytop}&=\{\tytop\}\cup\invgate{+}{\{\tytop\}},\\
 \invshort{-}{\beta}&=\{\beta\}\cup\invgate{-}{\{\beta\}},&
 \invshort{+}{\beta}&=\{\beta\}\cup\invgate{+}{\{\beta\}}\cup\invshort{+}{\Theta(\beta)},\\
 \invshort{-}{\garr{F}{G}}
   &=\garr{\invshort{+}{F}}{\invshort{-}{G}}\cup\invgate{-}{\garr{F}{G}},\\
 \invshort{+}{\garr{F}{G}}
   &=\garr{\invshort{-}{F}}{\invshort{+}{G}}\cup\invgate{+}{\garr{F}{G}}.
\end{align*}
Both operations distribute over choice. Universals have the same variance,
with lifted substitution and a body context with bound $\tytop$. Under local binders
chosen fresh for $R$, the lifted substitution fixes their variables and
all old variables, and maps the hole $\alpha$ to $R$ viewed in the extended
scope. Empty products are removed. An inhabited
positive structural branch already covers $\tytop$.

Induction over the construction shows that every old target member is
retained in the extended scope. Retraction makes each such member sound
after substitution.
The old-variable argument above, the fresh-hole tests and constructor
inversion establish the inverse specifications.

\subsection{Sharing family structure}

A family expression can be stored as an acyclic graph, with repeated
child expressions sharing a node. For the specialized inverse above,
take $R=\tytop\to\tytop$:
\[
 \invshort{-}{R\to R}=\garr{\{R,\alpha\}}{\{R,\alpha\}}.
\]
The two occurrences of $\{R,\alpha\}$ share one child node while allowing
independent choices, so the graph represents all four arrows.
Choice nodes retain their own children, preserving the branch correlations
of Section~\ref{sec:generators}.

Signed queries can use that structure directly. For an arrow target,
subtype inversion separates the two component tests:
\begin{align*}
 \qry{+}{\Theta}{\garr{F}{G}}{C\to D}
   &\iff\qry{-}{\Theta}{F}{C}\land\qry{+}{\Theta}{G}{D},\\
 \qry{-}{\Theta}{\garr{F}{G}}{C\to D}
   &\iff\qry{+}{\Theta}{F}{C}\land\qry{-}{\Theta}{G}{D}.
\end{align*}
Each conjunction chooses witnesses from the two children of the same
product. A choice node uses disjunction over its branches. Nonemptiness
is also structural: a product needs both children, a choice needs one,
and a singleton is inhabited. The test $\tytop\in\genlang{F}$ inspects
singleton and choice roots; constructor products contribute no $\tytop$ member.
These operations supply the empty-product and $\tytop$-membership tests
without materializing the represented types.

The single-hole inverse creates transformed constructor nodes and guarded
hole alternatives. Reusing a transformed child avoids repeating the
same inverse calculation. Appendix~\ref{sec:large-families} exhibits a
generator family with double-exponential cardinality and a linear shared
representation. Its members form an antichain, so dominance pruning
retains them all while sharing avoids their enumeration.

\subsection{Queries depend on the comparison context}

A query $\qry{p}{\Theta}{F}{Q}$ depends on the live family graph and
ordered bound assignments, together called a \emph{snapshot}. Replacing
a variable's bound $\tytop$ by its actual bound may change the answer.

Inverse substitution imports copies of variable bounds and their constructor
components at successive checking steps. A group fixes the scoped source
type, captured variable identities, scope cursor and query direction.
The cursor identifies a scope prefix; birth slots remain fixed on retirement.
Within one fixed snapshot and at a fixed endpoint, successes on live slots
form a prefix: older matching copies retain the alternatives of later ones.
A \emph{cutoff} records the first failure or an all-success marker.

Inverse guards are evaluated in the old snapshot before insertion mutates
the graph. Insertion and bound activation discard the prototypes' snapshot-local
comparison tables. Ordinary comparisons between variable-free types remain
valid across snapshots. Body selection retains its child family and retires
unneeded material without renumbering births.

For example, two inverse steps with replacement $R=\tytop\to\tytop$ and
successive fresh holes $\alpha,\gamma$ with bound $\tytop$ can give an earlier copy
$F_0=\{R,\alpha,\gamma\}$ and a later copy $F_1=\{R,\gamma\}$.
The negative query at endpoint $\alpha$ succeeds for $F_0$ and fails for
$F_1$, so the cutoff is one. Theorem~\ref{thm:backward} uses
the exact signed-query interface independently of this cache representation.

The evaluated variants differ in storage allocation. Dense initializes
scope-local comparison tables; shared allocates entries on demand and
also shares ordinary comparisons between variable-free types across
scopes. Both store the canonical domains attached to application targets
as family nodes.

\section{Mechanization and Evaluation}
\label{sec:mechanization}
\newcommand{\EvalPrograms}{130}
\newcommand{\EvalQueries}{582}
\newcommand{\EvalAccepted}{294}
\newcommand{\EvalRejected}{288}
\newcommand{\EvalBeyondCanonical}{16}
\newcommand{\EvalSameShapeRejected}{129}
\newcommand{\EvalTimedChecks}{12222}
\newcommand{\EvalCounterRuns}{1746}
\newcommand{\EvalReferenceDenseWins}{531}
\newcommand{\EvalReferenceSharedWins}{530}
\newcommand{\EvalMaxMedianMs}{131.2}
\newcommand{\EvalRepeats}{7}
\newcommand{\EvalBitsMin}{40}
\newcommand{\EvalBitsMax}{1805}

\newcommand{\EvalFamiliesMax}{1250}
\newcommand{\EvalEqualFamilies}{582}
\newcommand{\EvalSmallerComparisons}{260}
\newcommand{\EvalZeroComparisons}{322}

\subsection{Mechanization}
\label{sec:proof-correspondence}

The finite-family and backward checkers are terminating Rocq functions,
proved equivalent to the original typing judgments, including rejection
of malformed inputs. The backward checker is verified with both the general
inverse of Section~\ref{sec:inverse} and the single-hole specialization of
Section~\ref{app:single-hole-refinement}. Both satisfy the exact-inverse
interface, although they may produce different representatives.
We reuse the type syntax and substitution lemmas from the
POPLmark~\cite{AydemirEtAl2005} example distributed with
Autosubst~\cite{SchaeferTebbiSmolka2015}.
Autosubst uses functional extensionality.

An independent, axiom-free Rocq development verifies the execution model
in Appendix~\ref{app:cost} on parsed inputs: totality, unique answers and
charged work, equivalence to the original typing judgment, and a quartic
work bound. Decoding and the concrete representation of numeric fields
are analyzed on paper. The measured OCaml prototypes are independent
implementations; the cost certificate applies to the specified model.

\paragraph{LLM Assistance.}
All mathematical arguments and proof strategies were developed by the
authors. In the Rocq formalization, we used OpenAI Codex solely to assist with
completing proof scripts for selected lemmas. Rocq mechanically checked
the resulting proofs, and the authors reviewed the formalized statements.

\subsection{Prototype evaluation}
\label{sec:implementation-study}

We examine the distribution of symbolic family sizes and comparison storage
across complete checking judgments. Studies of Boolean-unification-based
inference likewise examine representation growth and query structure
\cite{MadsenVanDePol2020,MadsenVanDePolHenriksen2023}.

Three native OCaml checkers use machine integers and shared type nodes.
The \emph{finite-family reference} synthesizes a correlated symbolic family
and compares it with the target. The \emph{dense} and \emph{shared} backward
prototypes use the same scoped-cutoff procedure from
Section~\ref{sec:symbolic-implementation}, with the storage policies
described there.
The separate costed realization also uses producer-lineage transport and
persistent storage, which these prototypes do not implement.

\paragraph{Workloads and protocol.}
The fixed synthetic corpus combines call, composition and callback templates,
witness-carrying typed random generation, and stress families varying spine,
bound-chain and type depth. A nonminimum typing family exercises introductions,
dependent instantiations and mixed redexes. Scales are $k=2,4,8,16,32$.
After target deduplication, \EvalPrograms{} programs give \EvalQueries{}
distinct judgments: \EvalAccepted{} accepted and \EvalRejected{} rejected.
These include \EvalBeyondCanonical{} accepted targets outside the canonical
upper cone and \EvalSameShapeRejected{} rejections preserving the exposed
constructor. All generated programs are typable; rejected judgments test
target incompatibility. The reference assigns labels; an extracted
certified Rocq checker agrees on 231 judgments at scales two and four.

Input size $n$ counts the complete tuple $(\Delta,\Gamma,s,T)$ in a tree
encoding, including repeated occurrences despite implementation sharing.
The corpus spans \EvalBitsMin{}--\EvalBitsMax{} bits.
Each complete query starts with empty caches. Family-node counts cover all
symbolic family nodes created during the check, including those later
retired. All \EvalTimedChecks{}
uninstrumented checks and \EvalCounterRuns{} separate counter runs complete
and agree with the labels. Appendix~\ref{app:evaluation-details} gives the
encoding, generators, separate regression checks, protocol and full results.

\begin{figure}[htbp]
  \centering
  \includegraphics[width=\linewidth]{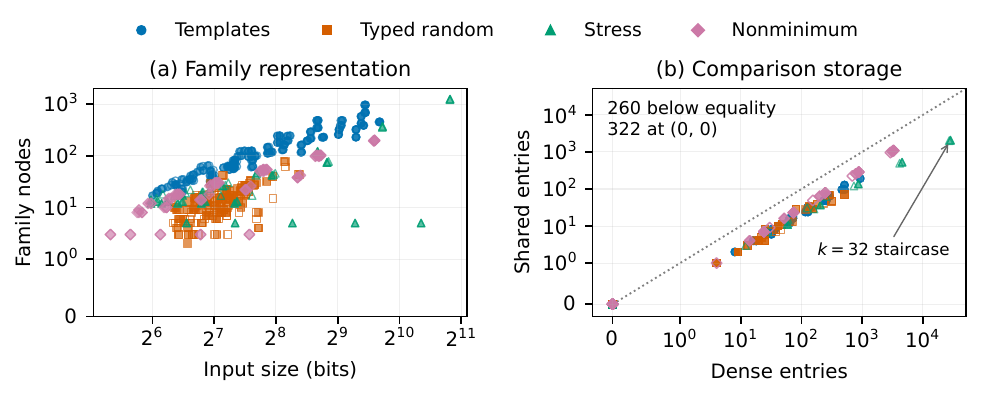}
  \caption{Structural cost across all \EvalQueries{} distinct judgments.
  (a)~Shared backward's interned family nodes against complete encoded input size;
  dense has identical family counts on all judgments.
  (b)~Paired dense/shared comparison entries; the dotted line denotes equality.
  Each point is one judgment, colored by workload family.
  Count axes retain zeros. Filled/hollow markers denote accepted/rejected targets.}
  \Description{Two scatter plots covering all 582 judgments, colored by four
  workload families. The first shows shared backward family nodes against
  complete input bits. Templates form an increasing band, while other families
  show different representation sizes at comparable input lengths. Stress
  includes both small representations and the maximum of 1250 nodes. The second
  pairs dense and shared comparison entries for each judgment: 260 pairs lie
  below equality and 322 are at the origin. The depth-32 staircase is annotated.}
  \label{fig:implementation-structure}
\end{figure}

\paragraph{Structural cost distribution.}
In Figure~\ref{fig:implementation-structure}(a), templates occupy an increasing
band, typed-random queries mostly use small families, and stress includes
both small representations and the largest families. At 149 input bits,
family counts range from 3 to 64. Both backward variants intern identical
counts on all \EvalEqualFamilies{} judgments, with a maximum of \EvalFamiliesMax{}.

\paragraph{Storage and runtime.}
Shared uses fewer entries on \EvalSmallerComparisons{}
judgments; the remaining \EvalZeroComparisons{} allocate none in either
variant (Figure~\ref{fig:implementation-structure}(b)). Dense counts include unused initialized cells; shared counts
include scoped and variable-free entries.

On the depth-32 staircase witness, both variants intern 1,250 family nodes,
while sharing reduces comparison entries from 28,569 to 2,081. Median CPU
time rises from 10.5 to 11.9\,ms and allocation volume from 10.2 to 11.4\,MB.
The reference has a lower CPU median than shared on \EvalReferenceSharedWins{} of
\EvalQueries{} queries. For the depth-32 dependent-instantiation program,
the median of its five query medians is 119.2\,ms for reference and
1.5\,ms for shared.

Appendix~\ref{app:evaluation-details} reports comparison entries by input
size and further measurements for the same prototypes.

\section{Consequences and Scope}
\label{sec:method-scope}

The inverse theorem applies to any fixed typed substitution. Identity
between different comparison contexts gives bound forgetting
(Section~\ref{sec:inverse}); weakening gives a second application, finite
type avoidance. Its use in another calculus requires that calculus's
own coverage and typing arguments.

\subsection{Finite type avoidance}
\label{sec:type-avoidance}

The inverse construction also computes types expressible after a local
name leaves scope, an avoidance problem studied for bounded existentials
\cite{GhelliPierce1998} and module signatures \cite{Crary2020}.
Let $B$ be well scoped in $\Delta_A$, and write
$\weaksubst:\Delta\hookrightarrow\Delta_A$ for the scope inclusion.

\begin{corollary}[Effective finite type avoidance]
\label{cor:type-avoidance}
For well-formed $\Delta,A$ and $B$ well scoped in $\Delta_A$, one can
compute finite families $F,G$ in $\Delta$ such that every well-scoped
$T$ in $\Delta$ satisfies
\begin{align}
 B\subin{\Delta_A}T
   &\iff \exists V\in\genlang{F}.\ V\subin{\Delta}T,
       \label{eq:upper-avoidance}\\
 T\subin{\Delta_A}B
   &\iff \exists V\in\genlang{G}.\ T\subin{\Delta}V.
       \label{eq:lower-avoidance}
\end{align}
\end{corollary}
\begin{proof}
Weakening preserves the old variable-bound comparisons, so $\weaksubst$
is typed. Take $F=\invup{\weaksubst}{B}$ and
$G=\invlo{\weaksubst}{B}$ in Theorem~\ref{thm:inverse}.
\end{proof}

All representatives avoid the local variable because they are well
scoped in $\Delta$. Equation~\eqref{eq:upper-avoidance} covers every
outer type that can serve as a supertype of $B$, without assuming a
single best choice. The lower family may be empty.

For example, take $\Delta=\beta<:\tytop$, $A=\beta$, and $B=\alpha$.
An exact upper cover is $F=\{\beta\}$: every outer supertype of $\alpha$ is
reachable through its bound $\beta$. An exact lower cover is empty, since
no type in the outer scope is below the fresh variable $\alpha$.

This corollary concerns subtype comparisons in \Ftop{}. A language with
local abstract types would need a typing rule that uses these
families.

\subsection{Generalization requirements}
\label{sec:extensions}

To transfer the decidability method to another calculus, one must establish
decidable subtyping on well-formed inputs, effective finite covers for its context changes,
and complete compositional operations for its term rules. Transferring the
backward algorithm requires a complete elimination interface
and finite inverse families for its substitutions. Runtime bounds require
a separate analysis of the chosen realization. Finite coverage alone
imposes no polynomial size bound.

\ruleref{All-Top} makes the body comparison context independent of candidate bounds.
We use this independence in both cone construction and the inverse
substitution step at type application.
These properties must be reconsidered when the quantified subtype rule changes.

Implicit instantiation introduces a further quantifier: the substitution
itself must be chosen. Greedy implicit bounded quantification provides a
decidable approach for a variant of kernel $\mathsf F_{<:}$
\cite{CuiJiangOliveira2023}. Our inverse theorem concerns each fixed typed
substitution. Taking the union of its solution sets over infinitely many
possible substitutions need not preserve finite generation. Extending
the method to implicit instantiation therefore requires a finite
representation of those choices and a complete procedure for querying it.

\section{Related Work}
\label{sec:related-work}

\subsection{Bounded quantification}
\label{sec:original-terms}

Bounded quantification combines parametric polymorphism and subtyping in
Fun \cite{CardelliWegner1985} and System $\mathsf{F}_{<:}$
\cite{CardelliEtAl1994}. Pierce proves undecidability of full subtyping and
typechecking \cite{Pierce1994}. Curien and Ghelli
\cite{CurienGhelli1992} establish minimum typing and a deterministic checking
procedure through normalization of typing derivations. Deterministic rule
selection does not imply that its subtyping calls terminate.
In \Ftop{}, subtype checking terminates,
but minimum typing fails \cite{CastagnaPierce1994,CastagnaPierce1995}.

The corrigendum proposes finite minimal-type families as a possible
route to checking. Our finite construction maintains coverage through
bound forgetting.

Laird's extended system recovers minimum typing using two
quantifiers \cite{Laird2023}. Typechecking of \Ftop{} inputs in that extension
is decidable (Proposition~8.5). Agreement on normal terms (Proposition~7.3)
gives decidability for the original normal fragment (Corollary~8.6).
Proposition~7.1 separates the original and extended judgments on arbitrary
terms.

Checking a normal form can accept a target that the original term lacks,
as Laird's counterexample shows. Take
$\Delta=\beta<:\tytop$, an empty term context, and
$T=\All{\alpha}{\beta}{\alpha\to \beta}$. The type-$\beta$ redex
\begin{align*}
  r &= (\tabs{\gamma}{\tytop}{\tabs{\alpha}{\beta}{\abs{x}{\gamma}{x}}})\{\beta\}
       \longrightarrow_\beta r',\\
  r' &= \tabs{\alpha}{\beta}{\abs{x}{\beta}{x}}
\end{align*}
satisfies
\[
 \typing{\Delta}{\varnothing}{r'}{T},
 \qquad
 \Delta;\varnothing\nvdash r:T
 \qquad\text{in \Ftop{}.}
\]
The positive judgment uses $\alpha<:\beta$ inside $r'$ to assign $\alpha\to \beta$ to
$\abs{x}{\beta}{x}$ before type abstraction. The negative judgment is
Laird's separating example. This failure of subject expansion prevents
direct reuse of the normal-form checker.
Our algorithms decide the original judgment on the unreduced input.

Barthe and Frade use finite minimal-type families for constructor subtyping,
where finiteness comes from a constructor's finite declarations
\cite[Section~4, Proposition~4]{BartheFrade1999}. In our construction,
representatives must cover the effects of bounds that body typing can use
but quantified subtyping discards. The generic finite-family checking
pattern is prior work. We establish effective completeness across the
change of comparison context.

Kernel $\mathsf{F}_{<:}$ compares universal types with equal bounds
\cite{CardelliWegner1985,Laird2023}. Strong Kernel $\mathsf{F}_{<:}$ uses two
contexts that retain the compared types' respective bounds. Hu and
Lhot\'{a}k establish decidable subtyping for it
\cite[Section~7]{HuLhotak2020}. Zhou et al.\ combine bounded quantification
with iso-recursive types, with decidable subtyping in their kernel variant
and undecidable subtyping in their full variant \cite{ZhouEtAl2025}.
These choices change the subtype relation. Our theorems concern the
single-context \ruleref{All-Top} rule, which our checking procedures retain.

\subsection{Target propagation and compact representations}

Bidirectional typing organizes synthesis and checking around the available
type information \cite{DunfieldKrishnaswami2021}. The backward equations
use this organization. Completeness in our calculus needs
the uniform elimination domains of Corollary~\ref{cor:backward-elimination}
and finite inverse substitution. In particular,
checking an argument against a canonical domain cannot be replaced by
comparing its canonical type with that domain.

MLsub combines ML-style inference and subtyping using a polar separation
and compact principal types, with simplification related to regular-language
algebra \cite{DolanMycroft2017}. Our signed families also preserve variance
and shared structure. Their meaning is a finite collection of ordinary
\Ftop{} types, and checking queries that collection under changing scoped
bounds. Their use requires the exact inverse-coverage theorem of
Section~\ref{sec:inverse}. A small representation alone does not establish
that property for a different type algebra.

Ghelli and Pierce \cite{GhelliPierce1998} address avoidance in Kernel Fun
with bounded existentials by computing a least variable-free supertype
and, when one exists, a greatest variable-free subtype.
Module calculi also address avoidance through existential signatures
\cite{Crary2020} or ZipML's floating fields, represented by zippers
\cite{BlaudeauEtAl2025}. Our
Corollary~\ref{cor:type-avoidance} instead gives finite covers of the outer-scope
supertypes and subtypes already expressible in \Ftop{}.

\section{Conclusion}
\label{sec:conclusion}

Finite covers under bound forgetting make typing bases compositional for
\Ftop{} and establish decidability for arbitrary original terms. Backward
checking uses rectangularity to fix elimination domains while preserving
every result type. Finite inverse substitution then transports acceptable
targets through type application. Both procedures are proved equivalent
to the original typing judgment in Rocq.

The inverse specification unifies bound forgetting and finite type
avoidance. Its exactness also permits different symbolic implementations.
Our experiments measure symbolic family sizes across synthetic checking
judgments and show how sharing reduces comparison entries.

\subsection{Future work}
\label{sec:future-work}

\paragraph{Local abstract types.}
The closest extension is a calculus with explicitly annotated bounded
packages, where unpacking introduces a hidden type name that must not
escape in the result type. Bounded existentials already expose this
avoidance problem~\cite{GhelliPierce1998}.
Corollary~\ref{cor:type-avoidance} suggests a compositional rule: replace
each member of the body's typing basis by its finite family of
outer-scope supertypes, then take their union. This would cover every
exportable result type without requiring a least avoiding type.
A conservative first language would keep package signatures separate
from the current grammar of ordinary types. Its new proof obligations
are package typing, substitution, and completeness of the scope-exit
rule. First-class existential types would additionally require new
subtyping and inverse-coverage cases. Extending the approach to ML
signatures~\cite{Crary2020,BlaudeauEtAl2025} would also require coverage
results for paths and sharing constraints.

\paragraph{Type schemes and correlated alternatives.}
Union and intersection types test both our finite-basis representation
and our rectangularity invariant. In a system with intersection
introduction, a nonempty finite sound complete basis would itself yield
a minimum type: the intersection of its members is assignable and lies
below every assignable type. Systems that still lack minima therefore
need a richer representation. Semantic subtyping provides such a
precedent: type schemes support decidable checking despite the failure
of minimum typing~\cite[Sections~5.4 and~6.12]{FrischCastagnaBenzaken2008}.
For backward checking, an overloaded type such as
$(A\to B)\land(C\to D)$ also ties different results to different domains;
one uniform domain need not preserve all results.
The finite arrow decomposition of
\citet[Lemma~6.37]{FrischCastagnaBenzaken2008} can be distributed into a
finite union of rectangles, providing a candidate replacement invariant.
The next question is whether symbolic schemes with these correlations
admit exact inverse queries under scoped substitutions. An explicitly
annotated, nonrecursive fragment would let us study closure and
termination before adding bounded quantification.

\paragraph{Implicit bounded instantiation.}
Greedy implicit bounded quantification~\cite{CuiJiangOliveira2023} and
higher-rank checking with unions and intersections~\cite{JiangCuiOliveira2025}
offer concrete calculi with carefully restricted instantiation rules.
Our inverse theorem handles a fixed substitution. Even for one candidate
quantified type $\All{\alpha}{A}{B}$, implicit instantiation must find an
$R$ satisfying both $R\subin{\Delta}A$ and
$B[R/\alpha]\subin{\Delta}T$. A complete solver must preserve both the
alternatives for $B$ and shared choices for $R$.

\bibliographystyle{ACM-Reference-Format}
\bibliography{references}

\appendix
\section{Supporting Proofs for Typechecking}
\label{app:cover-proofs}

\subsection{Proofs of the canonical elimination interface}
\label{app:rectangularity-proof}

Definition~\ref{def:structured-families} and
Lemma~\ref{lem:structured-rectangularity} state the structured-family
invariant used by Corollary~\ref{cor:canonical-rectangularity}.
We give its full proof and the typing induction establishing the invariant
for canonical types.

\paragraph{Typed substitutions.}
For $\sigma:\Delta\to\Theta$, the lift under a fresh binder $\alpha$ has
source $\Delta,\alpha<:A$ and target $\Theta,\alpha<:A[\sigma]$.
Subtype preservation follows by induction on the derivation: the bound
condition handles variables, and lifted substitutions handle the fresh
variable with bound $\tytop$ in \ruleref{All-Top}.

\paragraph{Universal body closure.}
For $(\mathcal S,B)$ structured in $\Delta_A$, reflexivity and bound restoration give
$\mathcal S\subseteq\upclosure{\Delta_\top}{\mathcal S}
\subseteq\upclosure{\Delta_A}{\mathcal S}$; transitivity then proves
\eqref{eq:structured-top-closure}. Since $B\in\mathcal S$, closure
replacement makes $(\upclosure{\Delta_\top}{\mathcal S},B)$ structured
under the actual bound. The same argument applies after substitution,
with actual bound $A[\sigma]$.

\begin{proof}[Proof of Lemma~\ref{lem:structured-rectangularity}]
Induct on the structured-family clauses, quantifying over every target
context $\Theta$ and typed $\sigma$.
For a singleton, exposure and subtype inversion give the equivalences.
The result family is the upper cone of the exposed body, using $\Theta$
for arrows and $\Theta,\alpha<:\tytop$ for universals. These families are
structured by the singleton and closure-replacement clauses, with
\eqref{eq:structured-top-closure} handling the universal case.

For an arrow family, substitution preserves its fixed domain and maps
the body family independently. Arrow inversion gives the result family
as the upward closure of that image, which is structured by the typed-image
and closure-replacement clauses.
For a universal family, lift $\sigma$ under its actual bound $A$.
The image body family is structured in $\Theta,\alpha<:A[\sigma]$.
\ruleref{All-Top} separates the candidate bound from the upward closure
of this family in $\Theta,\alpha<:\tytop$.
Equation~\eqref{eq:structured-top-closure} makes that closure structured
under the actual bound, with the substituted body as distinguished member.
Constructor inversion also rules out the other exposed head in each case.

For closure replacement, every member of either source family lies above
a member of the other. Subtype preservation transports these witnesses
through $\sigma$, so the images have equal upward closures as well.
The induction hypothesis transfers the factorization and exposed heads.
For a typed image, apply it to the composite typed substitution.
\end{proof}

\begin{lemma}[Canonical typing families]
\label{lem:canonical}
A term $s$ in well-formed contexts is typable exactly when there is a
$C$ such that $\canrel{\Delta}{\Gamma}{s}{C}$. This $C$ is unique and
assignable, and the pair
\[
 \bigl(\{T:\typing{\Delta}{\Gamma}{s}{T}\},
       C\bigr)
\]
is structured in $\Delta$.
\end{lemma}
\begin{proof}
First, induction on a canonical derivation proves assignability. Exposure
soundness types the function at its exposed type, and the declarative
argument guard justifies application. The other cases follow their typing
rules, with subsumption after exposure at type application.

For the converse, induct on a typing derivation, constructing both a
canonical derivation and the structured pair. Variables and the constant use singletons
followed by upward closure; introductions wrap the body family with the
annotated constructor and take upward closure for subsumption.

At application, the derivation supplies $f:U\to V$ and $u:U$.
Apply Lemma~\ref{lem:structured-rectangularity} with identity substitution
to the function family supplied by induction.
It gives an exposed arrow $A_0\to B_0$, a structured result family
$(\mathcal R,B_0)$, and $U<:A_0$, so subsumption gives $u:A_0$.
Together with the canonical derivation for $f$, this is precisely the
application rule's declarative guard and gives
$\canrel{\Delta}{\Gamma}{f\,u}{B_0}$.
Each $V\in\mathcal R$ is a result of $f$ at domain $A_0$; inversion places
every application typing above such a $V$. Hence the complete typing family
is $\upclosure{\Delta}{\mathcal R}$, structured with member $B_0$ by
closure replacement.

At type application, the same lemma gives an exposed universal
$\All{\alpha}{A_0}{B_0}$. Factorization and the derivation's bound test give
$R<:A_0$, satisfying the canonical type-argument guard and deriving
$\canrel{\Delta}{\Gamma}{f\{R\}}{B_0[R/\alpha]}$.
The body family $(\mathcal R,B_0)$ is structured in $\Delta,\alpha<:A_0$,
so $[R/\alpha]:(\Delta,\alpha<:A_0)\to\Delta$ is typed.
Each $V\in\mathcal R$ yields $V[R/\alpha]$ by \ruleref{Type-App} at $A_0$;
factorization retains the same $V$ from any type-application derivation.
Thus inversion and subsumption give the complete result family
$\upclosure{\Delta}{\mathcal R[R/\alpha]}$, structured with member
$B_0[R/\alpha]$ by typed image and closure replacement.
A final subsumption preserves the canonical derivation and leaves the
complete typing family unchanged.

Induction on canonical derivations gives uniqueness: lookup
fixes variable types, introductions wrap the unique canonical body type,
and eliminations use the unique canonical function type and deterministic
exposure. Consequently the structured pair constructed from any typing
derivation has the same distinguished $C$ as every canonical derivation.
\end{proof}

\subsection{Construction of general inverse families}
\label{app:inverse-construction}

Fix a typed $\sigma:\Delta\to\Theta$ between well-formed contexts. For a
family $F$ well scoped in $\Theta$, define the source-variable guards
\[
 \namedop{vars}^p_\sigma(F)
 =\{\alpha\in\domof{\Delta}:
                  \qry{p}{\Theta}{F}{\sigma(\alpha)}\}.
\]
The guards are decidable by finite-family queries and subtype decidability.

\paragraph{Singleton targets.}
Construct both signs together. The cases for $\tytop$ and variables are
\begin{align*}
 \invlo{\sigma}{\tytop}&=\{\tytop\},\\
 \invup{\sigma}{\tytop}&=\{\tytop\}\cup\namedop{vars}^+_\sigma(\tytop),\\
 \invlo{\sigma}{\beta}&=\namedop{vars}^-_\sigma(\beta),\\
 \invup{\sigma}{\beta}&=\namedop{vars}^+_\sigma(\beta)
                         \cup\invup{\sigma}{\Theta(\beta)}.
\end{align*}
For arrows, reverse the sign at the domain:
\begin{align*}
 \invlo{\sigma}{A\to B}
   &=\namedop{vars}^-_\sigma(A\to B)
       \cup\garr{\invup{\sigma}{A}}{\invlo{\sigma}{B}},\\
 \invup{\sigma}{A\to B}
   &=\{\tytop\}\cup\namedop{vars}^+_\sigma(A\to B)
       \cup\garr{\invlo{\sigma}{A}}{\invup{\sigma}{B}}.
\end{align*}
For a universal target, choose its binder $\alpha$ fresh for both contexts
and the images of $\sigma$. Write $\liftsub{\sigma}$ for the typed lift
\[
 \liftsub{\sigma}:
   (\Delta,\alpha<:\tytop)\longrightarrow(\Theta,\alpha<:\tytop),
\]
which fixes $\alpha$ and applies $\sigma$ to old variables. The equations are
\begin{align*}
 \invlo{\sigma}{\All{\alpha}{A}{B}}
   &=\namedop{vars}^-_\sigma(\All{\alpha}{A}{B})\\
   &\quad{}\cup\gall{\invup{\sigma}{A}}{\invlo{\liftsub{\sigma}}{B}},\\
 \invup{\sigma}{\All{\alpha}{A}{B}}
   &=\{\tytop\}\cup\namedop{vars}^+_\sigma(\All{\alpha}{A}{B})\\
   &\quad{}\cup\gall{\invlo{\sigma}{A}}{\invup{\liftsub{\sigma}}{B}}.
\end{align*}
The universal products bind the same fresh $\alpha$ in their body child.
Formation depends on variable scope, so a body formed with bound $\tytop$
is also formed under any generated bound.

Singleton recursion decreases $\tyweight{\Theta}{T}$
(Section~\ref{sec:termination}): following a variable bound or a
constructor child decreases weight, with an extension by $\alpha<:\tytop$ for universal
bodies. Each guard makes finitely many terminating subtype tests.

\paragraph{Target-family expressions.}
The empty and choice cases are
\[
 \invfam{p}{\sigma}{\varnothing}=\varnothing,\qquad
 \invfam{p}{\sigma}{F\cup G}
     =\invfam{p}{\sigma}{F}\cup\invfam{p}{\sigma}{G}.
\]
A singleton uses the preceding construction. For an arrow or universal
product, first test whether both child languages are nonempty. If either
is empty, return the empty family for both signs. Otherwise use
\begin{align*}
 \invlo{\sigma}{\garr{F}{G}}
   &=\namedop{vars}^-_\sigma(\garr{F}{G})
       \cup\garr{\invup{\sigma}{F}}{\invlo{\sigma}{G}},\\
 \invup{\sigma}{\garr{F}{G}}
   &=\{\tytop\}\cup\namedop{vars}^+_\sigma(\garr{F}{G})
       \cup\garr{\invlo{\sigma}{F}}{\invup{\sigma}{G}},\\
 \invlo{\sigma}{\gall{F}{G}}
   &=\namedop{vars}^-_\sigma(\gall{F}{G})\\
   &\quad{}\cup\gall{\invup{\sigma}{F}}{\invlo{\liftsub{\sigma}}{G}},\\
 \invup{\sigma}{\gall{F}{G}}
   &=\{\tytop\}\cup\namedop{vars}^+_\sigma(\gall{F}{G})\\
   &\quad{}\cup\gall{\invlo{\sigma}{F}}{\invup{\liftsub{\sigma}}{G}}.
\end{align*}
In the universal equations, $F$ has scope $\Theta$ and $G$ has scope
$\Theta,\alpha<:\tytop$; the result body has scope
$\Delta,\alpha<:\tytop$. Nonemptiness is structural: choice uses
disjunction, products use conjunction, and a singleton is nonempty.
Recursion follows the finite expression, using the singleton construction
at its leaves.

\begin{proof}[Proof of Theorem~\ref{thm:inverse}]
Prove formation, soundness and coverage simultaneously.
For every $V$ in the constructed lower family, there is
$T\in\genlang{F}$ with $V[\sigma]\subin{\Theta}T$; for every such $T$
and well-formed $Q$ with $Q[\sigma]\subin{\Theta}T$, there is a lower
representative $V$ with $Q\subin{\Delta}V$.
For the upper family, reverse both comparisons: members satisfy
$T\subin{\Theta}V[\sigma]$ for some target $T$, and
$T\subin{\Theta}Q[\sigma]$ supplies a representative
$V\subin{\Delta}Q$.

For singleton targets, induct on target weight over both contexts and every
typed substitution. Guards are sound by their tests and cover qualifying
source variables by reflexivity. Substitution preserves compound roots.

At target $\tytop$, every source type is below the lower representative $\tytop$.
For the upper direction, subtype inversion forces a nonvariable $Q$ whose
image is above $\tytop$ to be $\tytop$. The $\tytop$ and guard branches cover these cases.

At a target variable $\beta$, only source variables can have an image below
$\beta$, giving lower coverage. For upper coverage of a nonvariable $Q$,
algorithmic inversion of
$\beta\subin{\Theta}Q[\sigma]$ gives
$\Theta(\beta)\subin{\Theta}Q[\sigma]$, and the induction hypothesis for
the bound supplies a representative below $Q$. Its soundness follows from
$\beta\subin{\Theta}\Theta(\beta)$ and transitivity.

For target $A\to B$, a nonvariable source below the target must be an arrow
$Q=C\to D$. Inversion gives
$A\subin{\Theta}C[\sigma]$ and $D[\sigma]\subin{\Theta}B$.
Upper coverage of $A$ supplies $a\subin{\Delta}C$; lower coverage of $B$
supplies $D\subin{\Delta}b$. Hence
$C\to D\subin{\Delta}a\to b$, with $a\to b$ in the lower construction.
For the upper direction, an arrow source satisfies
$C[\sigma]\subin{\Theta}A$ and $B\subin{\Theta}D[\sigma]$.
The opposite child inverses supply $C\subin{\Delta}a$ and
$b\subin{\Delta}D$, so $a\to b\subin{\Delta}C\to D$.
The arrow rule combines child soundness in both directions.

For target $\All{\alpha}{A}{B}$, the compound source must be
$\All{\alpha}{C}{D}$. Inversion separates the bound comparison in $\Theta$
from the body comparison in $\Theta,\alpha<:\tytop$.
In the lower direction, child coverage gives
$a\subin{\Delta}C$ and
$D\subin{\Delta,\alpha<:\tytop}b$; in the upper direction it gives
$C\subin{\Delta}a$ and
$b\subin{\Delta,\alpha<:\tytop}D$.
The typed lift permits the body induction hypothesis in each case.
Applying \ruleref{All-Top} yields the required covering universal.
The same rule combines child soundness. In both constructor cases, guards
cover source variables, the upper $\tytop$ branch covers $\tytop$, and mismatched
compound roots are impossible. Child formation gives formation of the
products; replacing the bound $\tytop$ preserves body scope.

Induct next on the family expression. Empty families and empty products
have no witnesses; singleton leaves use the result above. Choices take
soundness and coverage from the branch containing the witness.
For inhabited products, child coverage supplies representatives as in the
constructor cases. Child soundness supplies independent target witnesses
whose combination belongs to the original product. Universal bodies use
the typed lift and \ruleref{All-Top}. Inhabitation justifies the positive
$\tytop$ branch. Formation of an inhabited product ensures formation of each
child's members; guards test source variables against the whole product.

Coverage gives one direction of \eqref{eq:inverse-family}; soundness,
subtype preservation under $\sigma$, and transitivity give the other.
\end{proof}

\subsection{Soundness and completeness of cone covers}
\label{app:cone-proof}

\begin{proof}[Proof of Theorem~\ref{thm:cones}]
Induct on the weight, proving well-formedness and representative soundness
together with both equalities. Soundness and restoration place each
right-hand set in \eqref{eq:upper} and \eqref{eq:lower} inside its left-hand
set by transitivity. For the reverse inclusions, expand membership in the
left-hand cone and cover the resulting subtype comparison.

For $\tytop$, the upper comparison forces the target to be $\tytop$. The lower
comparison is covered by $\tytop$ itself. For a variable $\alpha$, algorithmic
inversion of $\alpha<:T$ either gives $T=\alpha$ or a comparison from its bound
$C$ to $T$. The retained $\alpha$ covers the first case and the induction
hypothesis for $C$ covers the second. A subtype of $\alpha$ must be a variable,
so the lower variable-source family contains the source itself.

For $A\to B<:T$, either $T=\tytop$ or $T=C\to D$ with $C<:A$ and
$B<:D$. In the latter case lower coverage gives $a\in\Down{\Delta}{A}$
with $C<:a$ in $\forget{\Delta}$, and upper coverage gives
$b\in\Up{\Delta}{B}$ with $b<:D$ there. Hence the generated $a\to b$
is below $C\to D$ in the forgotten context. For $S<:A\to B$, a variable
$S$ is retained explicitly. Otherwise $S=C\to D$, with $A<:C$ and $D<:B$.
Upper coverage of $A$ and lower coverage of $B$ provide $a<:C$ and $D<:b$
in the forgotten context, so $C\to D<:a\to b$ as required.
The component soundness statements also give soundness of each generated
arrow in the original context.

For $\All{\alpha}{A}{B}<:T$, the case $T=\tytop$ is explicit. Otherwise inversion gives
$T=\All{\alpha}{C}{D}$ with $C\subin{\Delta} A$ and $B\subin{\Delta_\top}D$.
Choose $a\in\Down{\Delta}{A}$ and $b\in\Up{\Delta_\top}{B}$ with
$C\subin{\forget{\Delta}}a$ and $b\subin{\forget{\Delta_\top}}D$.
\ruleref{All-Top} then gives
\[
  \sub{\forget{\Delta}}{\All{\alpha}{a}{b}}{\All{\alpha}{C}{D}}.
\]
The representative is sound because $a\subin{\Delta} A$ and
$B\subin{\Delta_\top}b$. Lower coverage treats variable sources explicitly.
For a universal source, the induction hypotheses give an upper bound
representative and a lower body representative. Applying \ruleref{All-Top} again
gives the required comparison. All representatives are well formed by
component formation. Changing a bound does not change its variable scope.
\end{proof}

\subsection{Soundness and completeness of typing bases}
\label{app:typing-basis-proof}

\begin{proof}[Proof of Theorem~\ref{thm:typing-basis}]
Induct on $s$, generalizing over well-formed contexts $\Delta;\Gamma$.
If variable lookup fails, an abstraction annotation is ill formed, or a
type argument is ill formed, the construction returns the empty family.
Typing inversion excludes every assignable type in these cases, so both
soundness and coverage hold. Otherwise all context extensions used for
recursive calls are well formed. If a recursive basis is empty, the
induction hypothesis excludes a typing for that subterm; inversion excludes
one for the enclosing term, whose basis is also empty.
Every recursive call is on a proper subterm, and every auxiliary
construction terminates by Section~\ref{sec:termination}.
Typing inversion exposes the rule for that term constructor followed by a
possible final subsumption.

Variables and the constant use their declared types. At an arrow
introduction, the induction hypothesis gives a body representative below
the body type from inversion. Arrow covariance lifts that comparison,
and final subsumption reaches the requested target. Each generated arrow
is sound by the introduction rule.

At type abstraction, let inversion give a body type $R$ and a comparison
$\All{\alpha}{A}{R}\subin{\Delta} T$. Choose a body representative $B$ with
$B\subin{\Delta_A}R$. Corollary~\ref{cor:relative-cover} gives
$C\in\Up{\Delta_A}{B}$ with
\[
  \sub{\Delta,\alpha<:\tytop}{C}{R}.
\]
Hence \ruleref{All-Top} gives $\All{\alpha}{A}{C}<:_\Delta\All{\alpha}{A}{R}$ and then
$\All{\alpha}{A}{C}\subin{\Delta} T$. Soundness uses $B\subin{\Delta_A}C$ to type
the body before applying \ruleref{Type-Abs}. Cone coverage thus accounts for
the subsumption before type abstraction in Section~\ref{sec:obstacles}.

At an application, inversion gives $f:U\to V$, $u:U$, and $V\subin{\Delta} T$.
Choose a function representative $C\subin{\Delta} U\to V$. Equation~\eqref{eq:expose-arrow}
gives $\Expose{\Delta}{C}=A\to B$, $U\subin{\Delta} A$, and $B\subin{\Delta} V$.
An argument representative $D\subin{\Delta} U$ therefore passes the test
$D\subin{\Delta} A$, and the retained $B$ is below $T$.
Conversely, exposure is sound, and an argument representative passing
that test makes the application rule applicable with result $B$.

At a type application, inversion gives $f:\All{\alpha}{U}{V}$, $R\subin{\Delta} U$,
and $V[R/\alpha]\subin{\Delta} T$. By \eqref{eq:expose-all}, a function representative exposes as
$\All{\alpha}{A}{B}$ with $U\subin{\Delta} A$ and
$B\subin{\Delta,\alpha<:\tytop}V$. The bound test $R\subin{\Delta} A$ succeeds.
Preservation of subtyping by substitution from $\Delta_\top$ gives
$B[R/\alpha]\subin{\Delta} V[R/\alpha]$, so the retained type covers $T$.
Soundness follows from exposure, \ruleref{Type-App}, and the successful bound test.
All these arguments retain the final subsumption supplied by inversion.
\end{proof}

\section{Refinements of Finite Covers}
\label{sec:refinements}

\subsection{Minimal typing classes}

Fix a term $s$ and well-formed contexts $\Delta;\Gamma$ in this subsection.
An assignable type $M$ is minimal when every assignable $T<:M$ also satisfies
$M<:T$. Incremental insertion into a finite basis rejects a candidate
already covered by an existing member; otherwise it removes members covered
by the candidate and inserts it. Validity and coverage are preserved, with
one representative per mutual-subtyping class in the resulting antichain.

Upper cone pruning uses subtyping in $\forget{\Delta}$ and retains minimal
representatives. Lower cone pruning reverses that order and retains maximal
representatives. Term-basis pruning uses subtyping in $\Delta$.
Pruning an upper cone in the original context could delete a representative
needed after forgetting, so these orders are part of the specification.

\begin{theorem}[Complete minimal classes]
Every member of the computed term antichain is minimal, and every minimal
assignable type is mutually a subtype of some member of that antichain.
\end{theorem}
\begin{proof}
Let $B$ be a member and let an assignable $T$ satisfy $T<:B$. Coverage gives
a member $A<:T$, hence $A<:B$. The antichain property gives $B<:A$ (the two
members represent the same class), so $B<:T$. Conversely, a minimal $T$
has a member $B<:T$ by coverage. Assignability of $B$ gives $T<:B$.
\end{proof}

A typable term has a minimum type exactly when its antichain has one
equivalence class. Antichains can still be large
(Section~\ref{sec:large-families}).

\subsection{Large antichains and shared representations}
\label{sec:large-families}

In successive contexts $\alpha_0<:\tytop,\ldots,\alpha_i<:\tytop$, define
\[
 \mathcal A_0=\{\alpha_0\},\qquad
 \mathcal A_{i+1}=\{\alpha_{i+1}\}\cup
       \{(A\to\tytop)\to B:A,B\in \mathcal A_i\}.
\]
These families are antichains. For the compound members, two reversals of
domain variance give
\[
 (A\to\tytop)\to B<:(A'\to\tytop)\to B'
 \iff A<:A'\ \land\ B<:B'.
\]
A fresh variable bounded by $\tytop$ is incomparable with every compound member.
Induction therefore gives no duplicates and
\[
 c_0=1,\qquad c_{i+1}=1+c_i^2,\qquad c_{k+1}\ge 2^{2^k},
 \quad\text{where }c_i=|\mathcal A_i|.
\]
The family $\mathcal A_k$ has a direct shared representation with $4k+2$
nodes, counting the initial $\alpha_0$ and one shared node for $\tytop$. Each step adds
a variable, a choice, and two arrows, reusing the previous family.

This bound concerns the generator language; obtaining the same bound for
a small term's typing basis would require a realization proof.

\section{Evaluation Details}
\label{app:evaluation-details}

\subsection{Input encoding}
\label{sec:evaluation-input}

The input is the complete tuple $(\Delta,\Gamma,s,T)$. Its size $n$ is the
length of a tree encoding with separate de Bruijn indices for type and
term variables, two-bit type tags, three-bit term tags, and Elias-gamma
codes for $j+1$ for each natural field $j$~\cite{Elias1975}. Context lengths
precede their entries, followed by the term and target. Contexts use
newest-first order, with each type-variable bound scoped over older entries.
Repeated tree occurrences count separately even when their implementation
uses shared nodes. Thus $n$ includes contexts, annotations, explicit type
arguments and the target. The corpus spans \EvalBitsMin{} to \EvalBitsMax{}
bits; the generator parameter $k$ is not this encoding length.

\subsection{Workload constructions}

All families use scales $k=2,4,8,16,32$.

\paragraph{Templates.}
Templates start from $s_0=x$, with $x:A$, and repeat a type-preserving step
$k$ times. Call templates use $s_{i+1}=f\{A\}\,s_i$, with
$f:\All{\beta}{B}{\beta\to\beta}$ and $A\subin{\Delta}B$.
Composition uses $s_{i+1}=g\,(f\,s_i)$, with
$f:A\to(A\to A)$ and $g:(A\to A)\to A$.
Callbacks use $s_{i+1}=h\,f\,s_i$, with
$f:A\to A$ and $h:(A\to A)\to(A\to A)$.
In polymorphic calls, composition and callbacks, the context declares
$\alpha<:\tytop\to\tytop$ and $A$ ranges over
$\alpha$, $\alpha\to\alpha$ and $\All{\beta}{\tytop}{\beta\to\beta}$;
the polymorphic call bound is $B=\tytop$.
Bounded calls instead fix $A=\alpha$ and vary its declared bound $B$ over
$\tytop\to\tytop$, $\tytop\to(\tytop\to\tytop)$ and
$\All{\beta}{\tytop}{\beta\to\beta}$.
The target witness is $A$ in each case. A checking spine follows
abstraction bodies and function children; each application argument starts
a new spine. Iteration increases the number of spines while each spine has
length at most three.

\paragraph{Typed random.}
Typed random uses eight fixed seeds per $k$, ordered contexts of one to four
bindings, and type-generation fuel five. Witness-preserving choices cover
abstractions, environment-headed applications/instantiations, and term/type
redexes; unavailable environment choices fall back to redexes. The generator
is \texttt{xorshift64star-v1} with a recorded seed matrix; $k$ is the
recursive expansion budget.

\paragraph{Stress.}
Long spines apply $k$ consecutive type arguments to a function with $k$
nested quantifiers with bound $\tytop$ and innermost body $\alpha_k\to\alpha_k$, then
apply the result to $x:\tytop\to\tytop$. Every type argument is
$\tytop\to\tytop$; the longest stress spine has length 34.
Bound chains fix $C=\tytop\to(\tytop\to\tytop)$ and declare
$\alpha_1<:C$ and $\alpha_{i+1}<:\alpha_i$; with $x:\alpha_k$ and
$f:C\to C$, they check $f\,x:C$.
Type staircases grow quantified and arrow depth. Define
\[
 \begin{aligned}
 T_0&=\tytop,& T_{i+1}&=\All{\gamma}{\tytop}{T_i},\\
 S_0(A)&=\tytop,& S_{i+1}(A)&=\All{\gamma}{\tytop}{A\to S_i(A)},
 \end{aligned}
\]
with fresh binders. In
$\Delta=(\beta<:T_k,\alpha<:T_k)$ and $\Gamma=(x:S_k(\alpha))$,
the staircase query checks
$(\tabs{\gamma}{\tytop}{x})\{S_k(\beta)\}$ at $S_k(\alpha)$.
Its term has one type application while the surrounding types grow with $k$.

\paragraph{Nonminimum typing.}
In the context $\beta<:\tytop,\delta<:\tytop$, let
\[
 g_\beta=\tabs{\alpha}{\beta}{\abs{x}{\alpha}{x}},\qquad
 C_\beta=\All{\alpha}{\beta}{\alpha\to\alpha},\qquad
 P_\beta=\All{\alpha}{\beta}{\alpha\to\beta}.
\]
The canonical type is $C_\beta$, yet $g_\beta$ also has type $P_\beta$;
the All-Top rule prevents $C_\beta\subin{\Delta}P_\beta$.
Replacing the result $\beta$ in $P_\beta$ by the unrelated $\delta$
gives a rejected target with the same constructor skeleton.
Three variants per $k$ preserve this positive/negative pair.
Introductions wrap $g_\beta$ in $k$ term abstractions with $\tytop$ annotations,
and wrap its targets in matching arrows. Instantiations add $k$ nested type
abstractions, each bounded by the preceding variable (the outermost by
$\beta$), then instantiate all of them with $\beta$.
Mixed redexes additionally wrap the instantiated term in $k$ applications
$(\abs{z}{\tytop}{s})\,\tmconst$ with unused $z$.
These 15 programs test the failure of canonical subtype checking through
introduction and both elimination forms.

\paragraph{Targets and coverage.}
Programs contribute a witness, a distinct non-$\tytop$ structural supertype
when available, a local perturbation and an outer-shape change. For the
original families, a deterministic leaf search also retains the first
rejected target preserving the witness's exposed arrow or quantified head,
when one exists; it visits codomains and bodies first. The nonminimum family
instead supplies its explicit positive/negative pair.
Duplicate targets within a program are merged, retaining all construction
tags. The reference labels the resulting judgments before measurement;
terms retain redexes. A separate ordinary-subtyping diagnostic records
whether the canonical result is below each target. The corpus contains
\EvalSameShapeRejected{} same-head rejections and \EvalBeyondCanonical{}
accepted targets outside the canonical upper cone, including one from typed
random generation. All generated programs are typable; rejected judgments
test target incompatibility.

Separate Rocq regression proofs exercise nonminimum argument checking and
rejection at $\tytop$ of an abstraction with an untypable body. The OCaml regressions cover
malformed inputs, failed type-argument bounds, and nested type abstractions
and instantiations. These checks are outside the timed corpus.

\subsection{Measurement protocol and records}

The recorded run uses native OCaml 5.4.1 with \texttt{-O2} on an Apple M4
with 32\,GiB RAM.
Seven uninstrumented repetitions run in fresh native workers with empty
caches, sequentially with shuffled query and method order. Process CPU time
covers the complete check, including validation, synthesis and enabled GC.
Generation, target construction, startup, input registration, empty-engine
construction and the preceding full collection are excluded.
Input registration creates the common syntax pool. Empty-engine construction
also creates method-specific empty hash tables and their initial bucket
arrays: the reference's five tables and the backward engine's additional
synthesis/checking tables.
Spine construction, dense comparison arrays, on-demand shared entries and
their subsequent initialization are performed inside the timed check.

A separate counter run must agree with the timed runs. A ten-second
parent-enforced deadline includes untimed setup; each syntax or checker-node
pool is limited to 250,000 nodes. Timeout, node-limit, mismatch, allocation,
stack and worker failures remain explicit records. Medians require all seven
repetitions and the counter run to succeed.
The first and third quartiles linearly interpolate the sorted CPU samples
at positions $0.25(r-1)$ and $0.75(r-1)$ for $r=7$.
Table~\ref{tab:evaluation-outcomes} separates accepted and rejected queries
and reports the median relative interquartile range (IQR).

Check-call allocation uses \texttt{Gc.allocated\_bytes}, including probes;
peak RSS covers the whole worker. Logical comparison counts include unused
dense cells and on-demand shared entries.

The run archive contains the corpus, raw samples, counters and summary,
with environment, compiler profile, limits, plotting versions and SHA-256
hashes of the corpus, sources and executable. The reporter
checks corpus hashes and rejects duplicate rows, incomplete matrices and
label mismatches. The README gives measurement and export commands.

\subsection{Detailed measurements}

\begin{table}[htbp]
\centering\small
\caption{Synthetic OCaml workloads. Times are medians across completed queries of their median CPU times, in $\mu$s. Input lengths include both contexts and the target. Counts include accepted and rejected queries.}
\label{tab:implementation-workloads}
\begin{tabular}{lrrrrr}
\toprule
Workload & Queries & Input bits & Reference & Dense & Shared \\
\midrule
Polymorphic calls & 65 & 64--822 & 8.0 & 108.0 & 112.0 \\
Bounded calls & 60 & 68--541 & 8.0 & 106.0 & 111.0 \\
Composition & 65 & 86--700 & 8.0 & 98.0 & 100.0 \\
Callbacks & 65 & 86--700 & 7.0 & 64.0 & 69.0 \\
Typed random & 178 & 81--333 & 8.0 & 24.0 & 23.0 \\
Long spines & 20 & 63--462 & 11.5 & 87.5 & 89.5 \\
Bound chains & 25 & 63--196 & 8.0 & 20.0 & 20.0 \\
Type staircases & 25 & 94--1805 & 12.0 & 88.0 & 94.0 \\
Nonminimum typing & 79 & 40--775 & 44.0 & 40.0 & 41.0 \\
\bottomrule
\end{tabular}
\end{table}

\begin{table}[htbp]
\centering\small
\caption{CPU medians across accepted and rejected queries, in $\mu$s. Dispersion is the median per-query IQR divided by that query's median over 7 repetitions.}
\label{tab:evaluation-outcomes}
\begin{tabular}{lrrr}
\toprule
Method & Accepted & Rejected & Dispersion \\
\midrule
Finite-family reference & 9.0 & 8.0 & 19.0\% \\
Dense backward & 45.5 & 42.0 & 13.0\% \\
Shared backward & 45.0 & 41.5 & 13.6\% \\
\bottomrule
\end{tabular}
\end{table}

The reference has a lower CPU median on \EvalReferenceDenseWins{} of
\EvalQueries{} queries against dense and \EvalReferenceSharedWins{} against
shared. The largest per-query median is \EvalMaxMedianMs{}\,ms.
For the five targets of the depth-32 dependent-instantiation program,
the median of the five per-query medians is 119.2\,ms for reference and
1.5\,ms for each backward variant.

On the depth-32 staircase witness, both variants intern 1,250 family nodes,
while sharing reduces comparison entries from 28,569 to 2,081. Median CPU
time increases from 10.5 to 11.9\,ms and estimated allocation volume from
10.2 to 11.4\,MB. Fewer logical comparison entries therefore need not
reduce time or total allocation volume.

\begin{figure}[p]
  \centering
  \includegraphics[width=\linewidth]{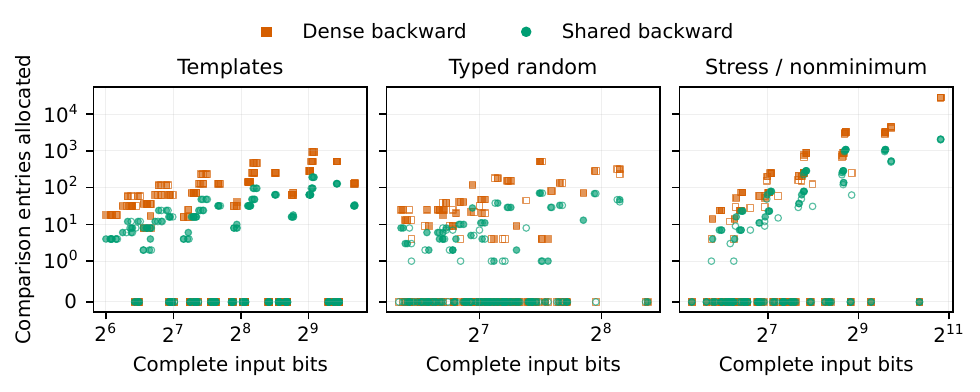}
  \caption{Logical comparison entries, including scoped ordinary and cohort
  queries. Dense counts include unused initialized cells; shared counts include
  on-demand scoped and variable-free entries. The symmetric logarithmic scale
  retains zeros. Filled/hollow markers denote accepted/rejected targets.}
  \Description{Three scatter plots of allocated comparison-cache entries
  against complete input length, with template, typed-random and combined
  stress/nonminimum panels. Shared backward usually allocates fewer entries than dense
  backward. The symmetric logarithmic vertical scale includes zero.}
  \label{fig:implementation-comparisons}
\end{figure}

\begin{figure}[p]
  \centering
  \includegraphics[width=\linewidth]{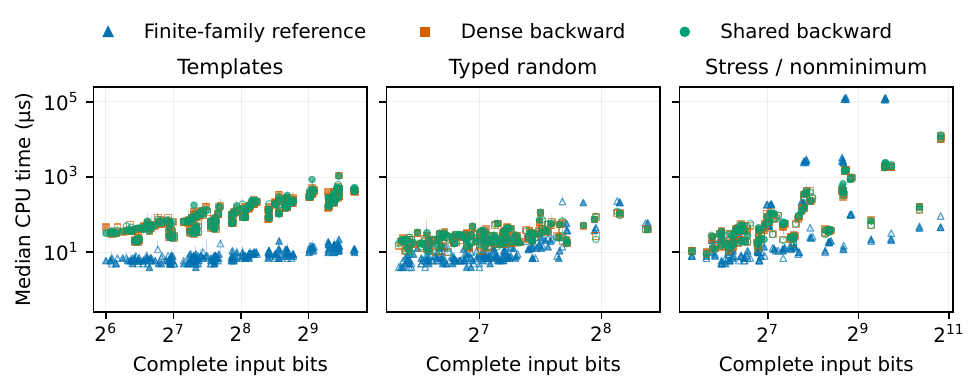}
  \caption{Median CPU time per judgment over seven uninstrumented checks
  in fresh workers; whiskers span the first and third quartiles.
  Filled/hollow markers denote accepted/rejected judgments.}
  \Description{Three scatter plots for template, typed-random and combined
  stress/nonminimum inputs. Horizontal axes show complete encoded input
  lengths; the shared vertical axis shows median CPU time on a logarithmic
  scale. The finite-family reference is below the backward implementations
  on most queries. Vertical whiskers show the interquartile range over seven
  samples. Filled and hollow markers distinguish accepted and rejected judgments.}
  \label{fig:implementation-time}
\end{figure}

\FloatBarrier

\section{Quartic Cost Bound}
\label{app:cost}
\label{sec:polynomial}

We analyze the checking equations of Section~\ref{sec:backward-equations}
with the single-hole inverse of Section~\ref{app:single-hole-refinement}.

\subsection{Execution and charges}
\label{sec:machine-model}
\label{sec:concrete-execution}

For $I=(\Delta,\Gamma,s,T)$, use the encoded length $n$ of
Section~\ref{sec:evaluation-input} and put $m=n+1$.
Words have $O(\log(n+2))$ bits; primitive word operations cost one.
Intern immutable descriptions in an ordered dictionary at logarithmic
cost; subsequent accesses use handles.

A checking spine follows abstraction bodies and function children;
application arguments start new spines. For spine lengths $\ell_i$ and
term-node count $\termnodes{s}$,
\begin{equation}
 \sum_i\ell_i=\termnodes{s}\le n,\qquad
 \spinesquares{s}=\sum_i\ell_i^2\le\termnodes{s}^2.
 \label{eq:spine-partition}
\end{equation}
Below, $\ell$ denotes one spine's length; $m$ always refers to the original
input, including in recursive argument checks.

\paragraph{Execution state.}
Each spine stores its current term occurrence, lexical context, saved
elimination domains, and a pointer to the target family graph of
Section~\ref{sec:symbolic-implementation}. A history $F$ is a mutable family
record with a stable identity, sign, scope cursor, fixed head, full constructor
children and root alternatives. Constructor frames retain the actual domain
and body pointers. A cursor identifies a scope prefix; binder tokens and
cursors keep their identities when insertion changes their numerical levels.
Ordinary type endpoints retain their captured tokens. Imported histories start
from ordinary types and attached producers from canonical types; these initial
types are their \emph{seeds}. Birth allocates a history and appends its slot
to its source group. Retirement never renumbers these slots.

\paragraph{Driver and spine operations.}
Validate all parsed syntax, including both children and empty context tails.
On success, build source indices and produce spine candidates with delayed
substitutions, saving domains and ordered scoped obligations. Reject if
candidate production fails. Otherwise discharge term arguments recursively
at their saved domains and test type arguments against
their saved bounds, in obligation order, stopping at the first failure.
Only then initialize the spine: allocate its lexical tokens, reserve deeper
scopes up to the target's maximum binder depth plus $4m$ using fresh tokens
with bound $\tytop$, and clone the target as a negative history.
Its query cache and hole-event list start empty.

The current term and target root determine the next action. Acceptance when
$\tytop\in\genlang{F}$ takes priority, as in Section~\ref{sec:backward-equations}.
Otherwise an application or matching introduction first reserves $4m$ further
comparison cursors and tokens with bound $\tytop$. Prepare the source directories,
child-map profiles and query tables for this snapshot, then perform the
corresponding operation:
\begin{enumerate}
\item At $f\,u$, attach the saved canonical domain $A$ as a positive family
and make an arrow frame with this domain and the complete current family $F$
as its body. Continue at $f$.
\item At $f\{R\}$, freeze the inverse guards in the current context, including
ordinary seed guards for imports needed in this pass. Insert the fresh
token $\alpha$ with bound $\tytop$, update histories and imports in
child-before-parent order, and append the alternatives selected by the frozen guards. Attach the
saved canonical bound $A$ and form the universal frame with the complete
inverse body $\invlo{[R/\alpha]}{F}$. Continue at $f$.
\item At $\abs{x}{A}{s}$, query the matching frame's positive domain against
$A$ in the current outer context. On success consume the frame, retain its
full body family, extend $\Gamma$ by $x:A$, and continue at $s$.
\item At $\tabs{\alpha}{A}{s}$, make the same domain query in the outer
context. On success consume the universal frame, activate its token by changing
its bound from $\tytop$ to $A$ at most once, and continue at $s$ with its full body
family in $\Delta,\alpha<:A$. Retire material that requires the consumed binder's
bound to remain $\tytop$.
\end{enumerate}
A failed domain guard terminates with rejection, reusing that query's result;
its reservation, preparation and query work have already occurred. The other
terminal cases prepare their query tables without a further scope reservation:
if $\tytop\in\genlang{F}$, accept; otherwise a variable $x$ returns
$\qry{-}{\Delta}{F}{\Gamma(x)}$, and a constant or an introduction without a
matching frame rejects. Every successful update consumes one spine node.

Preparation scans the retained group and view arrays, including retired slots,
to rebuild directories and select \emph{active rows}, those containing a live
history. It compiles profiles only for these rows, at shallow and deep cursors
alike, scanning their parent slots and child threshold ranges a constant number
of times. Import passes start from missing bound promotions of live histories
and close under ordinary-source child and bound dependencies.
Preparation also computes the core results needed by the operation's queries,
using Section~\ref{sec:compiled-cutoffs}; these are available for prepared-result
reuse. Queries read a fixed snapshot and evaluate dependencies before their
parent, with domains before bodies. A miss evaluates all requests of its compiled
equation in their stored order and caches the answer before the next request;
a hit reuses it. Before publishing a successful update, seal all new birth
records and, in the pre-update query snapshot, perform the seed comparisons
used to certify reusable cutoffs.
Publish the new root and eligible cache entries, then retire invalidated
regions. Compiled profiles and snapshot-local query tables are rebuilt at the
next preparation; deep comparison entries persist only under
Lemma~\ref{lem:deep-cache-preservation}.

A \emph{completed snapshot} consists of the live graph and ordered bound
assignments reached by this driver after initialization, comparison
reservation, or a successful update, with birth records sealed and retirement
finished. Reservation only
adds deeper tokens with bound $\tytop$; it leaves existing histories and bounds
unchanged. At cursor $c$, write $\Theta_c$ for this snapshot's comparison
context. All histories in a query are interpreted in their current versions
in that one snapshot, including the current bound assignments. Inverse guards
use the snapshot before insertion; queries on the resulting families use the
completed successor.

\begin{lemma}[Spine production]
\label{lem:spine-production}
Production and discharge of the saved obligations succeed exactly when
canonical synthesis succeeds at the spine root. On success, each saved
domain comes from the canonical exposure at its original application
occurrence, in that occurrence's lexical context. Target propagation reuses these
premises. Each term argument off the spine generates one recursive checking
obligation, with its original context and the saved domain.
\end{lemma}
\begin{proof}
Induct on term size. Production follows abstraction bodies and function
children, computing candidates on return. Lookup and introduction reproduce
the canonical rules. At $f\,u$,
expose the candidate for $f$, save its domain $A$, and append
$(\Delta,\Gamma,u,A)$ to the obligations; the candidate is its codomain.
At $f\{R\}$, save the exposed bound $A$ and the obligation
$R\subin{\Delta}A$, and apply the delayed substitution to the body.
Thus production is canonical synthesis with precisely these premises left
open. The induction hypothesis for the smaller argument terms makes their
full checks at $\{A\}$ equivalent to the saved typing premises. Failure of
production or an obligation therefore excludes canonical synthesis; success establishes
it at every position on the spine.

Subsequent checking visits the same occurrences and lexical contexts:
an introduction extends them by its written annotation, and an elimination
moves to its function child. Only the target family changes. Hence the saved
premises replace every repeated synthesis call in
Section~\ref{sec:backward-equations}, including the typability test before
accepting when $\tytop\in\genlang{F}$. Production emits each off-spine argument
once and target propagation emits none. The executed spines are consequently a subset of
the syntactic partition in~\eqref{eq:spine-partition}.
\end{proof}

\begin{definition}[Concrete run and charged work]
\label{def:costed-run}
$\namedop{run}(I,a,W)$ holds when the driver above terminates with answer $a$;
$W$ sums validation visits, setup and the charges below along its executed
operations, including recursive obligation checks.
Caches start empty. Charge unused initialization, discarded allocations and
work before rejection; invalid inputs incur only validation. Allocate one store
per spine initialization or type-argument check, numbered from zero in driver
order at three per store.
Spines and type-argument checks cost one plus their reached operations;
obligation traversal costs one per visited list node, including the empty tail.
\end{definition}

Query keys are (row, cursor, sign, endpoint, mode); misses evaluate the equations
of Section~\ref{sec:compiled-cutoffs}, children first. Hits, prepared-result reuse
and core insertion cost $1,2,4$, plus compilation and recursive work.
Shallow banks and shared columns retain their snapshot and bound versions.
Column keys are (cursor, sign, endpoint); hits and output reads cost $3,4$. A miss queries rows in
order, merges failures, traverses the forest and writes outputs for $11$
plus primitive-query work and counts of rows, aggregation/indexing events
and new cells.

Charge all retained slots visited during directory construction, and all
parent slots, child thresholds and token/cap entries read or initialized while
compiling active-row profiles. Directory capacity is
$1+\max(\#\text{groups},\#\text{variable views})$.
Reserving $4m$ comparison cursors and tokens with bound $\tytop$ costs
$8m+4m(2t+d)+6m(4m-1)$ for $t$ old tokens and $d$ deeper cursors.
Each store initializes $16m^2+1$ headers at two per header.
Publication reads and inserts new cells, then retires regions;
new-cell reads, writes, links and named-cell deletions cost $4,3,4,8$.

The mechanized model takes parsed syntax with abstract natural-number fields.
For an input of $b$ bits, a decoder scans the encoding once, using $b+1$ as
an overflow sentinel and rejecting oversized fields and malformed encodings.
Valid indices and context lengths are at most $b$, so this rejects no
well-formed input. Successful complete decoding preserves $n=b$ and leaves
fields that fit $O(\log(b+2))$-bit words. This linear decoding and numeric
representation argument is on paper.

\subsection{Coverage and source order}
\label{sec:histories}
\label{app:histories}
\label{sec:ordered-sources}
\label{sec:group-count}

$\histcore{F}$ keeps the fixed head and full children;
positive $\histbase{F}$ also keeps promotions. Lower inversion from a singleton
retains at most one structural root, so a matching introduction consumes one
frame and retains its full body after the domain test.

Application token $\alpha$ has saved guard $R_\alpha<:A_\alpha$ and positive
owning domain $\domainfam{\alpha}$ initially containing $A_\alpha$.
Its complete evolving domain, weakened along the live owner's scope-prefix
inclusion to observer $c$ in $\Theta_c$, is $\widehat D_{\alpha,c}$.
Pending added tokens have bound $\tytop$; nonself excludes $F$'s fixed variable head.

\begin{lemma}[Pending-owner preservation and root projection]
\label{lem:pending-owner-preservation}
In every completed snapshot reached after successful argument obligations,
each pending nonself root $\alpha$ in a live positive $F$ at $c$ has a live
owner and, for well-formed endpoints
in $\Theta_c$, satisfies
\begin{align}
 &\forall Q\text{ nonvariable}.\quad
 \qry{+}{\Theta_c}{\widehat D_{\alpha,c}}{Q}
 \Longrightarrow\qry{+}{\Theta_c}{\histbase{F}}{Q},
 \label{eq:pending-nv-certificate}\\
 &\forall Q.\quad
 \qry{+}{\Theta_c}{\widehat D_{\alpha,c}}{Q}
 \Longrightarrow\qry{+}{\Theta_c}{F}{Q}.
 \label{eq:pending-full-certificate}
\end{align}
These hold also inside live positive owning domains and survive birth, common
inverse, activation after a successful domain test, and retirement.
For nonvariable $Q$, live positive histories satisfy
$\qry{+}{\Theta_c}{F}{Q}\iff\qry{+}{\Theta_c}{\histbase{F}}{Q}$,
and negative histories satisfy
$\qry{-}{\Theta_c}{F}{Q}\iff\qry{-}{\Theta_c}{\histcore{F}}{Q}$.
\end{lemma}
\begin{proof}
Induct jointly on histories and owners. At birth, $U<:R_\alpha<:A_\alpha$
and upper-inverse completeness give structural or promotion witnesses for
nonvariable $Q$, and inhabitation for $Q=\tytop$. Inverse exactness and weakening
by declarations with bound $\tytop$ preserve them. Close strictly outward owner
dependencies transitively by depth before publication. At activation,
$V\in\genlang{\domainfam{\alpha}}$ and
$V<:A$ preserve the implications by weakening and composition, both at
observers and inside domains. Retain fixed literal and promotion branches
and original target roots; inner pending bounds stay $\tytop$. Activated tokens
leave the invariant; retirement deletes dependent regions. The implications
and subtype inversion give the projections.
\end{proof}
Source cohorts key imports by descriptor, cursor, captured tokens and sign,
and producers additionally by occurrence, role and exposure generation.
Annotation and declaration roles differ. Labels include inserted arguments;
exposure starts at the original bound. $\histcap{\mathcal C}$ counts allocated
slots, including retired births.
Here an occurrence is a position in the original type syntax or a
term-introduction wrapper. Substitution keeps surviving constructor labels
and labels inserted syntax by its explicit argument occurrence. Projection
selects a labelled child; exposure starts a new generation at the original
bound's root. Each attachment records the actual children, captured tokens
and owning wrapper. Its domain records retire when that wrapper is consumed;
the separate body continues.

\begin{lemma}[Source order]
\label{lem:source-order}
Fix a completed snapshot and a cursor with comparison context $\Theta$.
Let $F_i$ be the current history at birth slot $i$ in one source cohort or
attached producer group $\mathcal C$. All live $F_i$ use this same snapshot's
bound assignments. Their successes form a prefix at every well-formed $Q$
in $\Theta$:
\begin{equation}
 \qry{p}{\Theta}{F_i}{Q}\iff i<\kappa_{\mathcal C}^p(Q)
 \qquad\text{for live }i.
 \label{eq:cutoff}
\end{equation}
This holds for whole families and cores.
\end{lemma}
\begin{proof}
Imports retain their named seeds. A new copy of a seed is therefore covered
by an older matching history; applying the same exact inverse to both
preserves this implication. For producers, replay the result-contributing
synthesis path with its saved substitutions. Two live checkpoints of one
occurrence and generation satisfy $S_{b'}\sigma_{b,b'}=S_b$ for $b<b'$,
where $\sigma_{b,b'}$ erases the intervening restored slots. This identity
holds separately for each actual constructor child, with opposite sign on
domains and a lifted substitution in universal bodies. Inverse completeness
at each child supplies a witness in the older core, and hence whole-family
coverage as well. Labels and captured placements select the actual
substitution sequence.

While an owning wrapper is live, its seed binders have not been crossed.
Erasing later application tokens by their saved replacements remains typed
at activation: the saved argument guard and the successful domain test
give the substituted bound condition, also in pending owners by
Lemma~\ref{lem:pending-owner-preservation}. Substitution reflects an answer
back to the checkpoint, and narrowing preserves the converse. Fixed-head
comparisons use the same bound version on both sides. Retirement removes
positions without renumbering them, preserving the prefix on live slots.
\end{proof}
A newly computed cutoff is the first live failure or current capacity;
only $\infty$ certifies success at all future births. Deep reuse follows
Lemma~\ref{lem:deep-cache-preservation}.

\begin{lemma}[Canonical representation and producer grouping]
\label{lem:producer-groups}
Canonical and exposed types have $O(m)$ height and absolute scope, with scoped
DAGs of $O(m^2)$ nodes, $O(m)$ per level. At each snapshot, attached records
form $O(m)$ ordered producer groups per cursor and $O(m^2)$ in total,
with $O(m^2\ell)$ records cumulatively per spine.
\end{lemma}
\begin{proof}
For one candidate, follow its result path through bodies and function
children. Introduction adds a wrapper; substitution increases height by at
most the explicit argument's height. An argument term's synthesized type
is never inserted. Exposure resets the root to an original lexical bound,
so only the segment after the last reset contributes. Written annotations,
replacements and wrappers on that segment have total size at most $n$.
Height and lexical depth are therefore $O(m)$.

Memoization is local to this forward computation. During one substitution,
an old scoped node has one output. An inserted source occurrence is keyed
by its insertion level: its free tokens are fixed by the argument's lexical
capture, its local tokens by that placement. All copies undergo the same
remaining substitutions on the result path, so the key determines the
resulting scoped node. At final level $j$, an occurrence
of intrinsic depth $d$ was inserted at $j-d$; it contributes at most one
node there. Summing original occurrences gives $O(m)$ nodes per level and
$O(m^2)$ per candidate. Distinct archived candidates retain their own
records.

To count generations, consider two reloads of the same original bound on
this path. Bounds point strictly outward, so projections and bound lookups
alone cannot return to that occurrence. Let $e,i$ count the intervening
type eliminations and introductions. Returning to a variable root consumes
all $i$ introduced universal wrappers, and a return through a type source
consumes at least one more universal. Hence $e\ge i+1$. The restored root
level changes by $e-i>0$, and an occurrence at depth $d$ is restored at that
level plus $d$. Repeated live generations therefore occupy distinct cursors.
If the lexical boundary crosses an earlier placement, it has consumed its
owning wrapper and retired that domain's records. The unexposed declaration
copy adds only a separate constant role.

There are $O(m)$ occurrence/role pairs and at most a constant number of
live generations for each pair at a signed cursor. Canonical height and at
most $\ell\le n$ subsequent insertions confine attachments to $O(m)$
cursors. Each of at most $\ell$ attachments contributes $O(m^2)$ records,
giving the stated live and cumulative bounds.
\end{proof}

\subsection{Compiled queries and termination}
\label{sec:compiled-cutoffs}

For parent group $\mathcal C$ and a partial child-position map $\pi$ on
live parent slots, define
\[
 \namedop{pull}_\pi(k)=
 \min\bigl(\{i\in\domof{\pi}:\pi(i)\ge k\}
              \cup\{\histcap{\mathcal C}\}\bigr).
\]
Compile thresholds in one parent-plus-child scan. Import maps $\pi_d,\pi_b$
are nondecreasing between fixed cohorts and defined on all live parents,
so $i<\namedop{pull}_\pi(k)$ iff $\pi(i)<k$ for live $i$. Constructor cutoffs are
\begin{equation}
 \kappa_{\rm core}=\min(k_d,k_b),\qquad
 k_s=\namedop{pull}_{\pi_s}(\kappa_s)\quad\text{for imports}.
 \label{eq:cutoff-pull}
\end{equation}
Reverse the domain sign. Producer variable children may span several rows.

\paragraph{Shared variable columns.}
A prepared row $\eta$ of variable producers has snapshot-local key
(occurrence, role, exposure generation, cursor, sign);
referring parent births share its birth array
and fixed token $z$. For column $(c,p,Q)$ in comparison context $\Theta$,
$\mathcal R_z$ contains all such rows at $c$ with sign $p$ and head $z$,
even when their occurrence labels differ.
View $(\mathcal C,s)$ stores actual children $H_s(i)$ and birth
positions, masking nonvariable children and retired parents.

Number driver updates by their birth stamp $j$; initialization has stamp
zero and reservation does not advance the stamp.
The logical birth phase is $b(H)=3j$ for ownerless old-side
imports, $3j+2$ otherwise (including clones and owner-tagged imports);
token insertion is $t_z=3j+1$. Same-pass seed guards precede insertion
even if their imports are physically allocated later: they query the
ordinary seed in the pre-insertion context and the import receives that
pass's inverse. Clones are born from the restored source after the event.
Prepared whole-family cutoffs give
\[
 \tau(z)=\min\{b(H):\eta\in\mathcal R_z,\ H\text{ live in }\eta,\
                         \neg\qry{p}{\Theta}{H}{Q}\}.
\]
Empty minima use $\constructor{None}$, above current phases without certifying
future births. Saved variable replacements form a forest of earlier-token
parents; in token order compute
\[
 \vartheta(z)=
 \begin{cases}
   \vartheta(y),&\text{if }\namedop{parent}(z)=y\text{ and }\vartheta(y)<t_z,\\
   \tau(z),&\text{otherwise}.
 \end{cases}
\]
Views map phases to the first live parent whose variable child reaches them;
$v_s(Q)$ reads the view at $\vartheta(z_*)$ for last live path token $z_*$,
returning capacity if none qualifies. Construction costs
$O(\#\text{tokens}+\#\text{rows}+\#\text{views})$ numeric work plus primitive
queries; prepared lookups cost $O(1)$.

\begin{lemma}[Row and replay compatibility]
\label{lem:row-replay-compatibility}
Fix one completed snapshot and column $(c,p,Q)$, with $Q$ well formed
in $\Theta$. Every live $H$ in $\mathcal R_z$ satisfies $t_z\le b(H)$.
For live histories in these rows,
\begin{equation}
 b(H)\le b(H')\ \land\ \qry{p}{\Theta}{H'}{Q}
       \quad\Longrightarrow\quad\qry{p}{\Theta}{H}{Q}.
 \label{eq:row-phase-order}
\end{equation}
Each live variable child in a producer view is an entry of exactly one
prepared row. Its token lies on the replacement path ending at the view's
last live token. Its phase lies at or after its token's insertion and
strictly before the next insertion on this path, if any.
Moreover, if $\namedop{parent}(z)=y$, every live $y$-history in the prepared
rows born before $t_z$ covers every live history in the prepared rows for
$z$ or a descendant of $z$,
in the sense of the implication in~\eqref{eq:row-phase-order}.
\end{lemma}
\begin{proof}
Induct on updates, generalizing over well-formed endpoints.
At birth a variable producer has singleton seed $\{z\}$; its complete
family, including promotions, has that seed's signed cone.
Every older live history with fixed head $z$ retains $z$.
Hence it covers the newborn, regardless of their
occurrence labels. Two births in the same phase have the same singleton
cone. Common inverse preserves this coverage by exactness; activation
preserves it by the typed checkpoint erasure used in
Lemma~\ref{lem:source-order}. Induction on updates proves
\eqref{eq:row-phase-order}, including agreement at equal phases.

For the path assertion, restrict the saved-substitution replay to the
actual child of a fixed parent occurrence. A surviving constructor keeps
its label until replay restores a variable that supplied the whole child.
Thereafter a change from $y$ to $z$ can only undo an application whose
saved replacement is $y$; this is precisely the edge
$\namedop{parent}(z)=y$. The replay segment consists of the actual updates
between those attachment births, so this insertion lies strictly after
the earlier birth and before the later clone. Further undoing fixes old
tokens or restores a fresh one, and cannot return to an ancestor.
Thus the observations follow one path, possibly skipping tokens with no
live observation, and their intervals do not interleave. A child cannot
contain $z$ before its admission, giving the lower endpoint of each interval.
The child's source key and permanent birth slot identify its row entry;
different parents referring to it reuse this entry.

For the final assertion, an older fixed $y$-history succeeds at the literal
endpoint $y$. At the insertion with replacement $y$, inverse exactness
makes it succeed at $z$. This uses either the retained pre-event history
or its frozen ordinary seed guard; the latter has phase $3j<t_z$ even
when materialized later. A new variable producer at $z$ has the singleton
cone of $z$, so transitivity gives coverage. Iterate this argument along
saved-variable edges. Subsequent inverse and legal activation preserve it;
retirement merely removes observations. Post-event clones have phase
$3j+2>t_z$ and are not used as pre-event witnesses.
\end{proof}

\begin{lemma}[Producer route decomposition]
\label{lem:producer-route-decomposition}
Fix a reachable producer group $\mathcal C$, constructor side $s$, child
cursor, sign $p_s$ and endpoint $Q$. Each live parent $i$ has live actual
child $H_s(i)$; reverse the domain sign and observe universal bodies at
the shared cursor where the new binder has bound $\tytop$. Partition pointers
by variable/nonvariable fixed heads, preserving parent positions. The nonvariable route has one
source group $\mathcal D_s$ and map $\pi_s$ defined on those live parent slots.
The first variable-child failure is $v_s(Q)$, giving
\begin{equation}
 k_s(Q)=\min\bigl(\namedop{pull}_{\pi_s}(\kappa_{\mathcal D_s}^{p_s}(Q)),v_s(Q)\bigr).
 \label{eq:producer-route-decomposition}
\end{equation}
Missing failures contribute $\histcap{\mathcal C}$. After preparation, each
side uses at most one structural request, one profile lookup and one column read.
\end{lemma}
\begin{proof}
Backward replay preserves a nonvariable child's provenance until it restores
the variable that supplied that entire child. Exposure and projection occur
at the synthesis root, so cannot introduce another nonvariable source below
this surviving parent. This gives $\mathcal D_s$ and $\pi_s$.

By Lemma~\ref{lem:row-replay-compatibility}, merging row failures gives
the first failing phase across $\mathcal R_z$: a live $z$-history succeeds
exactly when its phase is below $\tau(z)$. In particular, a finite minimum
has a live failing witness; $\constructor{None}$ means all current entries
succeed. Induct along a view's replacement path. If
$\vartheta(y)<t_z$, its failing witness is at $y$ or an ancestor whose
failure already blocked the intervening edge. The lemma's blocking
implication makes all observations below $z$ fail.
Retain that threshold. Otherwise all preceding path observations succeed,
since their phases are below $t_z$; the new observations use $\tau(z)$,
whose finite value is at least $t_z$. This proves
\begin{equation}
 \begin{aligned}
  &\eta_i\in\mathcal R_{z_i},\quad
    \eta_i[a_i]=H_s(i),\quad b_i=b(H_s(i)),\\
  &\qry{p_s}{\Theta}{H_s(i)}{Q}
       \iff b_i<\vartheta(z_*).
 \end{aligned}
 \label{eq:variable-row-view}
\end{equation}
Here $\eta_i$ and $a_i$ are the actual child's row and permanent birth slot.
The strict order $3j<t_z<3j+2$ fixes the same-pass boundary. Every actual pointer
occurs in one mask, so component and side minima give the first failing
constructor, a prefix cutoff by Lemma~\ref{lem:source-order}.
\end{proof}
Whole-family unions take cutoff maxima. Each positive variable head retains
at most two promotion witnesses, one per bound version; compile their profiles
from actual references and applicable versions.

Seal $\birthrecord{F}{q}=\qry{+}{\Theta}{F}{R_q}$ at application-token
insertion, including same-pass guards; later or absent records are false.
Inverse exactness, retraction and \eqref{eq:pending-full-certificate}
preserve nonself answers. Nonself roots are application tokens, so
histories with arrow, universal or $\tytop$ heads use their record;
fixed-variable histories use
\[
 \qry{+}{\Theta}{F_z}{q}\iff
 \birthrecord{F_z}{q}\lor z\subin{\Theta}q.
\]
Negative $\rootholes{F}$ excludes the fixed head and retains old membership.
$\ancestors{\Theta}{q}$ contains $q$ and its variable-bound ancestors,
stopping at a nonvariable bound. Inversion gives
\[
 \qry{-}{\Theta}{F}{q}\iff\qry{-}{\Theta}{\histcore{F}}{q}\lor
 (\ancestors{\Theta}{q}\cap\rootholes{F}\ne\varnothing).
\]
Using $y\in\rootholes{F_i}\iff i<\mu_{\mathcal C}(y)$ on live births,
the hole cutoff is $\max_{y\in\ancestors{\Theta}{q}}\mu_{\mathcal C}(y)$
(missing entries zero). Only the activated token's row changes, since
outer bounds cannot mention it and inner pending bounds are $\tytop$;
each row has at most two versions.

For attached group $\mathcal C$, let $a_{\mathcal C}(q)$ count births before
$q$'s insertion; transport $R_q$ and $\Theta(q)$ to the query scope.
With $\chi_i(Q)=\qry{+}{\Theta}{F_i}{Q}$ and
$\psi_i(Q)=\qry{-}{\Theta}{F_i}{Q}$, use
\begin{align*}
 \chi_i(q)&\iff\namedop{seedvar}_i(q)\lor
       (i<a_{\mathcal C}(q)\land\chi_i(R_q)),\\
 \psi_i(q)&\iff\namedop{ownvar}_i(q)\lor\psi_i(\Theta(q))\lor
       (i<a_{\mathcal C}(q)\land\psi_i(R_q)).
\end{align*}
Here $\namedop{seedvar}$ tests $z<:q$ for fixed head $z$, and
$\namedop{ownvar}$ tests equality; both are false for nonvariable heads.
Omit absent replacements and zero caps. Inverse exactness at the live owner's
captured placement gives only earlier histories the replacement branch;
negative inversion gives the bound branch. Compile caps by minimum and
literal tests by greatest successful position plus one; the combined result
is a prefix by source order.

Order tokens by admission, declarations before pending tokens. Bounds
and replacements capture earlier tokens, so endpoint expansion along
arrow children, universal domains, bounds, replacements and retained $\tytop$
witnesses is acyclic when stopped at universal bodies. Let $e_c(Q)$ be
its height and $H_*$ the maximum level, including prepared sources.
Set $j=0,1,2$ for imports, variable views and producer or checking rows;
$h_{\rm src}(\eta)=e_c(S_\eta)$ for imported seed $S_\eta$, zero otherwise;
and $a=2,1,0$ for whole, fixed-variable helper and core modes. Use rank
\[
 \mathfrak r(\eta,Q)=
 (1,H_*-\namedop{level}(c),e_c(Q),j(\eta),h_{\rm src}(\eta),a)
\]
and $(0,0,\tyweight{\Theta}{S}+\tyweight{\Theta}{T},0,0,0)$ for ordinary
subtyping, with the weights of Section~\ref{sec:termination}.
Body calls decrease $H_*-\namedop{level}(c)$, endpoint expansions $e_c$,
dispatch $j$, imported promotion $h_{\rm src}$, and whole-to-core calls $a$.
Ordinary calls decrease the leading tag and then weight. Earlier coordinates stay fixed.
Rank induction proves termination and query exactness.

\subsection{Cumulative work and the quartic bound}
\label{sec:cost}

The \emph{source catalogue} has an instruction for $\tytop$ at rank zero and one
instruction per original type occurrence, including annotations, explicit
arguments, context declarations and the target. Compile older declarations
first and each type in postorder; rank is the instruction's catalogue index.
Arrow and universal instructions point
to their children; a variable records its written index, lexical depth and
the earlier instruction for its bound. A variable opened for
\ruleref{All-Top} points to rank zero. Thus every dependency has smaller
rank, and the catalogue size $b_{\rm cat}$ is at most $m$.
An imported instance retains its instruction, cursor, captured tokens and
sign; the capture interprets free variables and the cursor supplies local
binders. One pass memoizes this complete key.

\begin{lemma}[Source inventory and scope advance]
\label{lem:source-inventory}
The catalogue is closed under child and bound requests; saved replacements
refer to its original argument instructions. A pass starting at maximum
live level $h$ imports $O(m(h+m))$ heads, with $O(m)$ per level. The $4m$
comparison cursors reserved by the driver suffice for its construction
and queries. After $j$ spine updates, maximum live level is $O(m(j+1))$.
\end{lemma}
\begin{proof}
At an arrow, both children keep the cursor; at a universal, only the body
advances it by one. Bound lookup keeps the comparison cursor and restores
the bound's recorded capture. Each transition decreases instruction rank.
Consequently $\namedop{level}(c)+\namedop{rank}(r)$ never increases along
an import dependency, bounding imported levels by $h+b_{\rm cat}$.
A replacement starts from an existing catalogue root at its saved capture.

At a fixed cursor, an instruction has one live captured placement in a
pass. Its free tokens come from the saved source environment and its local
tokens from that cursor's prefix. A crossing of an old internal binder
discards paths using that placement: imports are fresh, hole edges do not
point back to an older import, and an imported constructor cannot become
the retained negative checking frame. In the deeper scopes no lexical
crossing occurs, so the placement stays fixed. Children, bound lookup and
replacement imports preserve this property with the captures just
specified. Thus the $O(m)$ instructions give $O(m)$ import keys per level,
including both signs, and the stated per-pass count.

For reservation, an old-side import adds at most $b_{\rm cat}$ levels;
insertion adds one. A canonical domain has binder depth at most $n$,
and completing its promotions adds at most $b_{\rm cat}$ more.
Taking the maximum over old material and the new domain bounds the
constructed levels by $h+n+2+b_{\rm cat}\le h+2m+1$.
An ordinary query from level $c$ descends at most $2n+1=2m-1$ further
levels: its endpoint syntax contributes at most $n+1$, and the ordered
lexical-bound catalogue contributes at most $n$; local comparison binders
have bound $\tytop$. Child requests decrease the remaining depth by one
whenever they advance the cursor. Thus construction followed by query
descent fits within $h+4m$. Initially the target depth and $4m$ padding
give the same room; a terminal query uses the room left by the last
update. Each step adds only $O(m)$ levels, proving the height claim.
\end{proof}

Call the original and canonical levels, padded by all remaining insertions,
shallow. Lemma~\ref{lem:producer-groups} bounds this band by $O(m)$;
a cursor's classification is fixed at admission. Deeper cursors contain
only imports, whose source-local tokens remain beyond every future lexical
boundary. Summing Lemma~\ref{lem:source-inventory} over the spine gives
\begin{align*}
 \#\text{cursors and tokens}&=O(m\ell),&
 \#\text{heads}&=O(m^2\ell^2),\\
 \#\text{shallow heads}&=O(m^2\ell),&
 \#\text{deep cohorts}&=O(m^2\ell).
\end{align*}
The last count uses the $O(m)$ fixed instruction/capture keys at each of
$O(m\ell)$ deep cursors across all passes; further matching births append
to those cohorts.
Attach deep events $[0,k)$ at $k-1$ in a birth chain reachable from earlier
births: $g\le\ell$ events use $O(\histcap{\mathcal C}+g)$ storage.
Shallow masks use one slot per head per stage. Links target live or fresh
objects; retirement prevents resurrection. Each node and edge has one allocation
and live interval, giving $O(m^2\ell^2)$ amortized graph work by reference counting,
including discarded records.

\begin{lemma}[Deep-cache preservation]
\label{lem:deep-cache-preservation}
Let a deep imported cohort have seed $S$, current capacity $B$, and an
accurate cutoff for endpoint $Q$ in whole-family or core mode. A persistent
key fixes the cohort, cursor, sign, mode, captured endpoint and applicable
bound versions. Its source-local tokens remain beyond all future lexical
boundaries, and its free lexical tokens have final bounds. Publication tests
$S\subin{\Theta_c}Q$ for positive sign and $Q\subin{\Theta_c}S$ for negative sign.
If this seed test succeeds, the tag $\infty$ is valid for future births.
If it fails, the current finite cutoff $k\le B$ remains valid.
These tags survive reservation, common inverse, frame operations and
activation while their cursor and captures satisfy the stated conditions.
An all-success answer without a seed certificate is snapshot-local.
\end{lemma}
\begin{proof}
Transport $S$ and $Q$ by their stable captured tokens.
Reservation only appends deeper binders with bound $\tytop$. At inverse,
the retraction replacing the fresh token fixes these old endpoints;
inverse exactness therefore preserves the answers of retained births.
For core mode use the same identities at the actual constructor children;
a negative core at a variable endpoint also retains that endpoint's bound
version. Frame operations leave retained imported histories unchanged.
Activation cannot cross a deep source's local binders or change a captured
free bound. Later application tokens in the evolving family may activate;
erase them by their saved replacements. The successful argument and domain
guards make this erasure typed by
Lemma~\ref{lem:pending-owner-preservation}. Since it fixes $S$ and $Q$,
substitution reflects any new witness to the earlier query, while narrowing
preserves the converse. Thus activation preserves both retained answers
and the ordinary seed test.

Every new matching birth is a copy of $S$, or its inverse in a pass whose
fresh token is absent from the old $Q$. Its answer at $Q$ is exactly the
seed test, by inverse exactness (childwise for cores). If that test succeeds,
every old birth retains the seed and every future birth succeeds, justifying
$\infty$. If it fails, each new position $i\ge B\ge k$ fails. Retained old
positions keep their answers, so the same predicate $i<k$ remains exact.
The predicate continues to hold when the first failing birth retires:
slots are permanent, it concerns only live positions, and the seed
certificate still determines future births.

Success at every current live position gives cutoff $B$.
Publication must perform the seed test: success certifies $\infty$;
failure makes $B$ a persistent finite cutoff whose future births fail.
Without that test, current capacity or $\constructor{None}$ cannot be reused
after growth. Retirement removes entries whose cursor or capture is no
longer admissible; a changed bound version has a different key. Such entries
are never revived under a reused numerical level. Ordinary deep comparisons
obey the same token and bound stability, without the birth cases.
\end{proof}

\begin{lemma}[Cumulative comparison capacity]
\label{lem:comparison-capacity}
Comparison regions initialize $O(m^3\ell)$ cells per spine; after profile
and shared-column preparation, each consumer uses $O(1)$ local work and
makes $O(1)$ recursive requests.
\end{lemma}
\begin{proof}
A shallow snapshot has $O(m)$ cursors, $O(m)$ groups and import rows per cursor,
$O(m^2)$ frame and target rows and $O(m)$ endpoints per cursor. Imports have one
live captured placement; negative frames have no bound promotion. Source size,
intrinsic depth and live owners bound closure under children, bounds and
replacements. Over $O(\ell)$ snapshots this gives $O(m^3\ell)$ cells.

By Lemma~\ref{lem:deep-cache-preservation}, each eligible deep cell is
initialized and evaluated once across snapshots, even if its witnessing
birth later retires. Lemma~\ref{lem:source-inventory} gives $O(m)$ named
endpoints per deep cursor, with fixed captures. Thus the $O(m^2\ell)$
deep cohorts have $O(m)$ persistent columns each. Ordinary rectangles have
$O(m)$ sources on each side at $O(m\ell)$ cursors and obey the same bound. Constant modes,
bound versions and reserves include unused cells and retired snapshots.
Lemma~\ref{lem:producer-route-decomposition} bounds prepared work and branching;
updates pay for repeated roots and hits.
\end{proof}

Every live import is reached by promotions and constructor edges from an
initial history, checking frame or producer at level $O(m)$
(Lemma~\ref{lem:producer-groups} and the input bound). Promotion keeps the
cursor. Along imported dependencies, level plus catalogue rank cannot increase,
and rank is at most $m$ (Lemma~\ref{lem:source-inventory}). Hence live imports
also have level $O(m)$. With $O(m)$ source keys per cursor and two child views
per producer group, there are $O(m^2)$ active rows.

Closing missing promotions under these dependencies creates $O(m^2)$ imports
per memoized pass. A constant number of passes, one canonical attachment and
one frame add $O(m^2)$ records per update. After $j$ updates, the retained
array has $O(m^2(j+1))$ records, including retired ones. Each record contributes
one group slot and each parent at most two child-view slots; nonempty directories
obey the same bound. Radix grouping and directory construction scan fixed-length,
polynomial-range keys for these slots in $O(m^2(j+1))$ work per snapshot.

Each active row has $O(j+1)$ parent slots and a constant number of child
threshold ranges of that length: a pass contributes at most one birth per
cohort, an attachment one per producer group. With $O(m)$ token/cap entries per
row, rebuilding profiles costs $O(m^2(j+1)+m^3)$ per snapshot, including deep
rows. The cumulative directory and profile work is
\[
 O\!\left(\sum_{j=0}^{\ell}\bigl[m^2(j+1)+m^3\bigr]\right)
   =O(m^2\ell^2+m^3\ell).
\]
By Lemma~\ref{lem:row-replay-compatibility}, the row key and permanent slot
recover every actual variable child; referring parent births create no
additional rows. Lemma~\ref{lem:producer-groups} therefore bounds these
rows and the two views per parent group by $O(m^2)$ per snapshot, each
assigned once to a cursor and sign. Token
widths are $O(m)$ per cursor (source-depth window, lexical tokens and $\ell$
applications), $O(m^2)$ per snapshot. The $O(m)$ endpoint columns cost
$O(m^3\ell)$ numeric work; primitive requests use the counted cache cells.

Owner columns and twice-built ancestor rows across $O(m\ell)$ cursors cost
$O(m^3\ell)$; remapping, scans and requests cost
$O(m^2\ell^2)$. Interning $O(m^2\ell)$ immutable descriptions costs
$O(m^2\ell\log(m+2))$; mutable copies reuse handles.
Canonical production and type-argument setup cost $O(m^3\ell)$ per spine.
Total spine work is $O(m^3\ell+m^2\ell^2)$.

\begin{theorem}[Quartic charged work and RAM realization]
\label{thm:quartic}
For every parsed $I$, the model relation $\namedop{run}$ of
Definition~\ref{def:costed-run} has a run with unique answer and charged work,
accepting exactly when $\typing{\Delta}{\Gamma}{s}{T}$ holds.
One input-independent constant $\runconst$ satisfies
\begin{equation}
 \begin{aligned}
 \inputledger{n}{s}&=n+(n+1)^3\termnodes{s}+(n+1)^2\spinesquares{s},\\
 W(I)&\le\runconst\,\inputledger{n}{s}.
 \end{aligned}
 \label{eq:runtime-ledger}
\end{equation}
Its charged work is $O((n+1)^4)$. With the decoding and numeric representation
argument of Section~\ref{sec:machine-model}, the word-RAM realization takes
$O((b+1)^4)$ time on a bit-string input of length $b$.
\end{theorem}
\begin{proof}
Induct on term size for recursive obligations and on completed updates
within each spine. Initialization establishes the history invariants and
an empty cache. At an established snapshot, query-rank induction proves
the current guards accurate. Use those witnesses to establish pending-owner
coverage, replay compatibility and cache preservation at the successor.
Reservation preserves the old state and supplies scope room by
Lemma~\ref{lem:source-inventory}. Thus every reached query and update
terminates; fixed operations and consecutive stores give uniqueness.
Lemma~\ref{lem:spine-production} replaces repeated canonical calls, and
exact queries and frames implement Theorem~\ref{thm:backward}, giving
typing equivalence. Summing spine work using \eqref{eq:spine-partition},
with validation and setup charged once, gives \eqref{eq:runtime-ledger}
and the quartic charged-work bound. Decoding adds $O(b)$ work; on success
$n=b$, giving the stated word-RAM bound.
\end{proof}

\end{document}